%% file: main.tex
\documentclass[a4paper,11pt,onecolumn,nopdfoutputerror]{quantumarticle}
\usepackage[utf8]{inputenc}
\usepackage[english]{babel}
\usepackage[T1]{fontenc}
\usepackage{amsmath,amsthm,amsfonts,amssymb}
\usepackage{hyperref}
\usepackage{xurl}
\usepackage{braket}
\usepackage{booktabs}
\usepackage{graphicx}
\usepackage{tikz}
\usetikzlibrary{shapes.misc}
\usepackage[numbers,sort&compress]{natbib}

\newtheorem{theorem}{Theorem}
\newtheorem{lemma}[theorem]{Lemma}
\newtheorem{proposition}[theorem]{Proposition}
\newtheorem{definition}[theorem]{Definition}

\theoremstyle{remark}

\newtheorem{example}[theorem]{Example}

\newcommand{\F}{\mathbb F}

\newcommand{\Zn}{\mathbb Z_n}
\newcommand{\wt}{\operatorname{wt}}
\newcommand{\one}{\mathbf 1}

\begin{document}

\title{Short qudit triorthogonal codes from doubling}

\author{Shiroman Prakash}
\affiliation{Department of Physics and Computer Science, Dayalbagh Educational Institute, Agra-282005, India}
\email{sprakash@dei.ac.in}
\author{K. L. Ramya}
\affiliation{Department of Physics and Computer Science, Dayalbagh Educational Institute, Agra-282005, India}
\email{ramya144401@dei.ac.in}

\maketitle

\begin{abstract}
We construct short triorthogonal codes for qudits of odd prime dimension by generalizing
doubling, which combines a triorthogonal code with a code admitting a transversal phase gate, and raises the distance by up to two.
This gives codes with parameters $[[28,1,4]]_3$, $[[37,1,5]]_3$, $[[72,1,6]]_3$,
$[[118,1,8]]_3$, $[[33,1,5]]_5$ and $[[20,1,4]]_7$, each shorter than the cyclic, concatenated
and Reed--Muller codes of the same distance, and a doubling construction of the
$[[15,1,3]]_3$ code. When the columns of a generator matrix are regarded as charged points,
triorthogonality becomes the vanishing of multipole moments up to octupole order, and doubling
becomes a double layer of opposite charges; this picture explains why inputs with nontrivial
coefficients can give shorter codes. We also show by exhaustive search that the shortest codes
of distance three have length $11$ for $p=5$ and length $7$ for $p=7$. Combined with high-rate
distance-two codes, the new codes reduce the cost of qutrit magic state distillation by up to
an order of magnitude in regimes with high input noise.
\end{abstract}

% =============================================================================
\section{Introduction}\label{sec:intro}
% =============================================================================

Triorthogonal codes~\cite{Bravyi_2012} are CSS codes on which a transversal third-level
diagonal gate acts as a non-Clifford logical gate, and they underlie low-overhead magic state
distillation~\cite{Bravyi_2012,Haah_2018,Hastings_2018}. For qubits, the most effective way to build
such codes with growing distance is \emph{doubling}. Introduced by Betsumiya and
Munemasa~\cite{Betsumiya_2012} as a map from doubly-even to triply-even classical codes, its quantum
version combines a code with transversal Clifford gates and a code with a transversal $T$ gate into
a longer $T$-code of larger distance. The $[[15,1,3]]$ and $[[49,1,5]]$ codes arise in this
way~\cite{BravyiCross2015}, as does the $[[95,1,7]]$ code~\cite{Sullivan2024}, and Jain and Albert
used doubling systematically to obtain the shortest known codes at larger
distances~\cite{Jain_2025}.

Magic state distillation extends to qudits of odd prime
dimension~\cite{ACB,CampbellAnwarBrowne,campbell2014enhanced,Prakash:2024ghq}, but small triorthogonal codes are
scarce in this setting; the standard examples are punctured
Reed--Muller~\cite{CampbellAnwarBrowne,campbell2014enhanced,Saha:2025frb} and
Reed--Solomon~\cite{Krishna_2019} codes, the distance-two codes of~\cite{Prakash:2025azi}, and the
cyclic codes of~\cite{PrakashCyclic}. Other recent work builds binary CSS-T and triorthogonal
codes from cyclic, matrix-product and self-dual
codes~\cite{CampsMoreno2024,Camps-Moreno:2024qcn,Berardini2024,JaramilloVelez2026,Baldelli2026,ShiLuKimSole2024}.
Asymptotically good codes with transversal non-Clifford
gates~\cite{wills2024constantoverheadmagicstatedistillation,golowich2024asymptoticallygoodquantumcodes,nguyen2024goodbinaryquantumcodes,SanJose2026},
and distillation with asymptotically good qudit codes~\cite{Cervia2025}, target large lengths,
whereas we focus on small codes.

Prakash and Saha noted that no small qutrit triorthogonal code of distance three, analogous to
the 15-qubit code, was known~\cite{Prakash:2025azi}. Such a code, $[[15,1,3]]_3$, is presented
in~\cite{prakashGeneration}, where it is obtained from one generation of the Standard Model,
described as a doubled code, and shown to be the unique shortest qutrit triorthogonal code of
distance at least three. Here it appears as the simplest case of the general doubling
construction and in the early rounds of qutrit distillation sequences. It is not self-orthogonal, so it is triorthogonal only under the definition
of~\cite{PrakashCyclic,SahaNotes}, which drops that requirement; we use this definition throughout.

Doubling extends to qudits once the self-dual input of qubit doubling is replaced by a
\emph{$P$-code}: a self-orthogonal code, orthogonal to the all-ones vector, which therefore
admits a transversal phase gate $P$ (Section~\ref{sec:pcodes}). The phase gate, rather than the
qudit Hadamard, is the relevant Clifford gate because the symmetry of the magic state we
distill is a phase-type Clifford (Section~\ref{sec:magic}). Doubling a triorthogonal seed $[[n_T,1,d_T]]_p$, which need
not be self-orthogonal, with a $P$-code $[[n_P,1,d_P]]^P_p$ gives a
$[[2n_P+n_T,1,\min(d_P,d_T+2)]]_p$ code (Theorem~\ref{thm:doubling}), and among tensor-product
constructions doubling is the most economical (Appendix~\ref{app:blowup}).

Section~\ref{sec:geometry} explains why doubling works. Each column of a generator matrix is a
point carrying a charge in $\F_p$, triorthogonality is the vanishing of multipole moments up to
octupole order, and doubling places two copies of the $P$-code in layers of opposite charge,
which raises that order by one. Coinciding points can be merged into one point carrying the sum of their charges, which leads to weighted $P$-codes. For $p=3$ and $p=5$ at distance three and $p=7$ at distance four these give the shortest codes, including a doubling construction of $[[15,1,3]]_3$.

Doubling, iterated where useful, gives codes of distance up to $8$ that are shorter than the
shortest codes of the same distance in~\cite{PrakashCyclic} (Section~\ref{sec:results} and
Table~\ref{tab:best}); their generator matrices are in Appendix~\ref{app:matrices}. An exhaustive search (Section~\ref{sec:search},
Appendix~\ref{app:search}) shows that the shortest codes of distance three have length $11$ for
$p=5$ and $7$ for $p=7$, complementing length $15$ for $p=3$~\cite{prakashGeneration}, and that
no code of distance four for $p=5$ is shorter than $15$. In qutrit distillation
(Section~\ref{sec:sequences}), using the new codes in the first rounds lowers the cost by up to
an order of magnitude at high input noise. Doubling does not help everywhere: at low input noise
the optimal sequences do not use the new codes, and for $p=5$, for $p=7$, and for $p=3$ at
distance $7$, several shortest known codes remain cyclic or concatenated (Table~\ref{tab:best}).

% =============================================================================
\section{Preliminaries}\label{sec:prelim}
% =============================================================================

Throughout, $p$ is an odd prime, $\omega=e^{2\pi i/p}$, and $\one$ is an all-ones vector. For
$a,b\in\F_p^n$ let $\beta(a,b)=\sum_la_lb_l$ and, for coefficients $\mathbf c\in(\F_p^\times)^n$,
$\tau_{\mathbf c}(a,b,e)=\sum_l\mathbf c_la_lb_le_l$.

\subsection{Qudit magic states}\label{sec:magic}

For $p\ge5$, we consider the non-Clifford cubic gate
$U=\sum_j\omega^{j^3}\ket j\bra j$, and for $p=3$ the gate $U=\operatorname{diag}(1,\zeta,\zeta^{-1})$
with $\zeta=e^{2\pi i/9}$~\cite{HowardVala}.  Injecting $U$ consumes one copy of the magic state
\[
\ket M=U\ket+,\qquad \ket+=p^{-1/2}\sum_{j=0}^{p-1}\ket j.
\]
This is the state qudit triorthogonal codes are designed to distill.

The gate $U$ lies in the third level of the Clifford
hierarchy~\cite{Cui}. Conjugating the Pauli operator $X$ by $U$ gives
$UXU^\dagger\ket j=\omega^{3j^2+3j+1}\ket{j+1}$, which is of the form
\[
UXU^\dagger=XD\qquad\text{with $D$ a diagonal Clifford gate}.
\]
For $p=3$, this also holds with $D=\zeta\operatorname{diag}(\omega^{j^2+j})$. Because $X\ket+=\ket+$, the magic state is
an eigenstate of the Clifford $XD$~\cite{jain2020qutrit}, and $Z^k\ket M$ are the other eigenstates, with
distinct eigenvalues. 

Averaging a noisy copy over the powers of $XD$ makes it diagonal in the
basis $Z^k\ket M$, so that the noise becomes a random $Z^k$ error~\cite{ACB,CampbellAnwarBrowne}. 
This is why only the $Z$-distance is relevant for distillation. For simplicity, we assume $Z^k$ is equally likely for each $k \neq 0$ in what follows, and describe the noise by its
infidelity $\epsilon$, the total probability of a nontrivial error $Z^k$. An alternative noise
parameter is the depolarizing strength $\delta$: depolarizing noise replaces the state by the
maximally mixed state with probability $\delta$, which gives $\epsilon=(p-1)\delta/p$. Neither noise parameter is canonical
for comparison across qudits of different dimensions. We use $\epsilon$ for the properties of individual codes, where it is the natural
variable, and $\delta$ when comparing qubits with qudits, following~\cite{campbell2014enhanced,Prakash:2025azi}. At equal
$\delta$ a qudit has the larger error probability, so an advantage of qudits that survives at equal
$\delta$ is more robust.

For qubits, $T\ket+$ is Clifford-equivalent to
an eigenstate of the Hadamard gate~\cite{MSD}, so the same magic state can be distilled with codes
that have a transversal Hadamard, and the self-dual codes used in qubit doubling have one. For
qudits of dimension $p>2$, this coincidence no longer holds. The eigenstates of the Fourier gate, the qudit Hadamard, are
Clifford-inequivalent to $\ket M$~\cite{jain2020qutrit}, and codes built around them, such as the
ternary Golay and quantum quadratic-residue codes, distill different Clifford
eigenstates~\cite{2020golay,Zurel2026}. The symmetry of $\ket M$ is instead the phase-type Clifford
$XD$, so the natural input for doubling is a code with a transversal \emph{phase} gate
$P=\sum_j\omega^{j^2}\ket j\bra j$ (Section~\ref{sec:pcodes}), whose qubit analogue is the gate $\operatorname{diag}(1,i)$. Our codes implement $U^\lambda$ rather than $U$. When $\lambda$ is a cube (as it is for every code in
this paper), $U^\lambda\ket+$ is Clifford-equivalent to $\ket M$~\cite{PrakashCyclic,HowardVala,Bengtsson_2014}.

\subsection{Triorthogonal matrices}

We follow the conventions used by one of us earlier~\cite{PrakashCyclic}. A matrix $H\in\F_p^{m\times n}$ with independent rows, the first
of which is the logical row $y$ and the rest of which span the stabilizer space $\mathcal A$,
defines a CSS code with $X$-stabilizers $X(\mathcal A)$ and $Z$-stabilizers
$Z(\mathcal H^\perp)$, where $\mathcal H$ is the row space. Its distance is the $Z$-distance
$d=\min\{\wt(v):v\in\mathcal A^\perp\setminus\mathcal H^\perp\}$, which by
Section~\ref{sec:magic} is the one relevant for distillation, and which for triorthogonal codes is
also the distance of the code~\cite[Lemma~6]{PrakashCyclic}. The vectors of
$\mathcal A^\perp\setminus\mathcal H^\perp$ are the nontrivial \emph{$Z$-logicals}.

\begin{definition}[\cite{PrakashCyclic}]\label{def:tri}
The pair $(H,\mathbf c)$, with logical row $y$, is \emph{triorthogonal} if
\begin{enumerate}
\item $\tau_{\mathbf c}(h^\alpha,h^\beta,h^\gamma)=0$ for all rows $h^\alpha,h^\beta,h^\gamma$, not
necessarily distinct, except $h^\alpha=h^\beta=h^\gamma=y$; and
\item $\lambda:=\tau_{\mathbf c}(y,y,y)\neq0$.
\end{enumerate}
A subspace on which $\tau_{\mathbf c}$ vanishes identically is a \emph{triorthogonal space}.
\end{definition}

The transversal gate $\bigotimes_lU^{\mathbf c_l}$ then implements $U^\lambda$, exactly for
$p\ge5$ and up to Pauli corrections for $p=3$~\cite{PrakashCyclic}.  We write a triorthogonal code as $T=(y;T_0)$, call it a \emph{seed}, and write $\lambda_T$
for its logical exponent. The trivial seed $[[1,1,1]]_p$ has $y=(1)$ and $T_0=0$. Note that, unlike earlier work~\cite{Krishna_2019}, no self-orthogonality is required~\cite{PrakashCyclic, SahaNotes} for a transversal non-Clifford gate.

\subsection{$P$-codes}\label{sec:pcodes}

The doubling procedure needs a triorthogonal seed, and a code on which the phase gate $P$ acts transversally. The
simplest such codes are the following. For qudits of $p>2$, doubling also allows a weighted generalization that appears in
Section~\ref{sec:masses}.

\begin{definition}\label{def:pcode}
A \emph{$P$-code} of length $n$ is a subspace $S\subseteq\F_p^n$ such that
\[
S\subseteq S^\perp,\qquad S\subseteq\one^\perp,\qquad n\not\equiv0\pmod p .
\]
It defines the CSS code with $X$-stabilizers $X(S)$, $Z$-stabilizers $Z\big((S+\F_p\one)^\perp\big)$
and logical $\bar X=X(\one)$, with distance
\[
d_P=\min\{\wt(v):v\in S^\perp,\ \textstyle\sum_lv_l\neq0\}.
\]
We write $[[n,1,d_P]]_p^{P}$ for its parameters.
\end{definition}

\begin{lemma}\label{lem:P}
For a $P$-code, $P^{\otimes n}$ implements the logical gate $P^n$.
\end{lemma}
\begin{proof}
A code state is a uniform superposition of $\ket{x\one+s}$ with $s\in S$, and
$\beta(x\one+s,x\one+s)=x^2n+2x\sum_ls_l+\beta(s,s)=x^2n$.
\end{proof}

$P$-codes play the role of the inputs of qubit doubling. For $p=2$ the formula for $P$ gives
$Z$, so for qubits we take instead $P=\operatorname{diag}(1,i)$, the $S$ gate. The computation
in the proof of Lemma~\ref{lem:P}, now with the phase $i^{\wt(x\one+s)}$ and weights counted
modulo $4$, shows that $P^{\otimes n}$ implements $P^n$ when $n$ is odd and every vector of $S$
has weight divisible by $4$. These doubly even codes of odd length are the inputs of qubit
doubling~\cite{Jain_2025}. When $\dim S=(n-1)/2$, as for the Steane code,
$(S+\F_2\one)^\perp=S$, so the $X$- and $Z$-stabilizers coincide and the code also has a
transversal Hadamard, which is why qubit distillation is often phrased in terms of Hadamard
eigenstates~\cite{knill2004faulttolerant,Reichardt2005}. For $p>2$, doubling uses only the
conditions of Definition~\ref{def:pcode}, and the $P$-codes need not have dimension $(n-1)/2$.

The condition $n\not\equiv0 \pmod p$
makes the logical gate nontrivial, and is needed to solve~\eqref{eq:kappa} below. Since $S$ is
totally isotropic in the nondegenerate space $\one^\perp$ of dimension $n-1$, a $P$-code has
$\dim S\le(n-1)/2$, that is, $n\ge2\dim S+1$. The $P$-codes we use, from cyclic and
quadratic-residue codes and from a search, are collected in Section~\ref{sec:inputs}.

% =============================================================================
\section{Doubling}\label{sec:doubling}
% =============================================================================

The construction takes two copies of a $P$-code with opposite signs, adds a \emph{linking row}
that ties the second copy to the logical row of the seed, and appends the seed. The logical row
runs over all three blocks.

\begin{theorem}[Doubling]\label{thm:doubling}
Let $S$ be a $P$-code of length $n_P$ with distance $d_P$, and let $T=(y;T_0)$ be a seed of
length $n_T$ with distance $d_T$. Let $\eta,\kappa\in\F_p^\times$ satisfy
\begin{equation}\label{eq:kappa}
\eta\,n_P\kappa^3+\lambda_T=0 .
\end{equation}
Then the matrix
\[
H=\left(\begin{array}{ccc}
-\kappa\one & \kappa\one & y\\ \hline
s & -s & 0\\
0 & \kappa\one & y\\
0 & 0 & T_0
\end{array}\right)\qquad(s\text{ in a basis of }S),
\]
with coefficient $\eta$ on the first $2n_P$ coordinates and the seed's own coefficients on the last $n_T$, is triorthogonal with
$\lambda=\lambda_T$, and defines a $[[\,2n_P+n_T,\ 1,\ \min(d_P,\,d_T+2)\,]]_p$ code. For $p=3$ or
$p\equiv2\pmod3$ one can take $\eta=1$, since cubing is then a bijection of $\F_p^\times$; for $p\equiv1\pmod3$ one can
take $\kappa=1$ and $\eta=-\lambda_T/n_P$, so that the doubled code carries nontrivial coefficients.
\end{theorem}

\begin{proof}
\emph{Independence.} In a vanishing linear combination of the rows, the first block forces the
coefficients of the logical row and of the rows $(s,-s,0)$ to vanish, since $\one\notin S$
(because $\sum_l1=n_P\neq0$ while $S\perp\one$); the second block then forces the coefficient of
the linking row to vanish, and the third block those of the rows $(0,0,t)$.

\emph{Triorthogonality.} Write $L$ for the
logical row and $K=(0,\kappa\one,y)$ for the linking row, and sort the triples of rows by which
rows they contain.
\begin{itemize}
\item \emph{Triples not containing $K$ or $(0,0,t)$.} On the first $2n_P$ coordinates, $L$ and the rows $(s,-s)$
all have the form $(x,-x)$, and the two halves contribute with opposite signs, so these triples
vanish there; on the seed coordinates only $\tau(L,L,L)=\lambda_T$ survives.
\item \emph{Triples containing $K$ and two rows $(s,-s)$:} $\eta\kappa\,\beta(s,s')=0$, since $S$ is self-orthogonal.
\item \emph{Triples containing $K$, one row $(s,-s)$, and $K$ or $L$:} $-\eta\kappa^2\sum_ls_l=0$, since $S\perp\one$.
\item \emph{Triples containing only $K$ and $L$, and at least one $K$:} $\eta\kappa^3n_P+\lambda_T=0$, by~\eqref{eq:kappa}.
\item \emph{Triples containing at least one row $(0,0,t)$:} only the seed coordinates contribute, where every other row is $0$
or $y$, so these triples vanish because $T$ is triorthogonal.
\end{itemize}

\emph{Distance.} Write a $Z$-type operator as $Z=(z_A,z_B,z_C)$, with blocks of lengths $n_P$,
$n_P$ and $n_T$, and write $\sum z_A$ for the sum of the coordinates of $z_A$. It is a $Z$-logical
when it commutes with all stabilizer rows, which gives one condition for each type of
row:
\begin{itemize}
\item the rows $(s,-s,0)$ give $\beta(s,z_A-z_B)=0$ for all $s\in S$, that is, $z_A-z_B\in S^\perp$;
\item the rows $(0,0,t)$ give $z_C\in T_0^\perp$;
\item the linking row gives $\kappa\sum z_B+y\cdot z_C=0$.
\end{itemize}
It is nontrivial when it does not commute with the logical row, which is a logical $X$, that is, when
\[
-\kappa\sum z_A+\kappa\sum z_B+y\cdot z_C\neq0.
\] By the linking condition, this is equivalent to
$\sum z_A\neq0$.

There are two cases.
\begin{itemize}
    \item If $y\cdot z_C=0$, then $\sum z_B=0$ by the linking condition, so
$z_A-z_B\in S^\perp$ has coordinate sum $\sum z_A\neq0$, and it is a $Z$-logical of the $P$-code. Its
support lies in the union of the supports of $z_A$ and $z_B$, so $\wt(Z)\ge d_P$.
    \item If
$y\cdot z_C\neq0$, then $z_C\in T_0^\perp$ is a $Z$-logical of $T$, so $\wt(z_C)\ge d_T$; moreover
$\sum z_B=-y\cdot z_C/\kappa\neq0$ and $\sum z_A\neq0$, so $z_A$ and $z_B$ are both nonzero and
$\wt(Z)\ge d_T+2$.
\end{itemize}  
Both bounds are attained, by $(v,0,0)$ with $v$ a minimum-weight $Z$-logical of the
$P$-code, and by $(ae_1,ae_1,z_C)$ with $z_C$ a minimum-weight $Z$-logical of $T$ and
$a=-y\cdot z_C/\kappa$.
\end{proof}

The above procedure is a natural generalization of doubling for qudits of dimension $p>2$. Let us note a few details of this construction, relevant mainly for $p>2$.

Note that the relative sign between $(-\kappa\one,\kappa\one)$ in the logical row and $(0,\kappa\one)$ in the
linking row is essential: with the opposite sign, $\tau(L,K,K)=2\lambda_T\neq0$. For $p=3$ the rule is
$\kappa\equiv-\lambda_T/n_P\pmod3$.

Only triorthogonality of the seed is used in the proof, so seeds need not be self-orthogonal; the
cyclic $[[19,1,3]]_3$ and $[[11,1,3]]_5$ codes of~\cite{PrakashCyclic}, for instance, are valid
seeds. Other choices of $(\eta,\kappa)$ than those in the theorem also work, and Appendix~\ref{app:matrices} uses
$\eta=2$, $\kappa=3$ for $[[20,1,4]]_7$.

Doubling is a special case of a tensor-product construction, in which each coordinate of a small matrix is replaced by a copy of a $P$-code (Appendix~\ref{app:blowup}). Within that family,
$m$ copies of the $P$-code gain at most $m$ units of distance, and only $m=2$ gains exactly $m$
(Proposition~\ref{prop:optimal}): tripling and its relatives exist but are less efficient.

% =============================================================================
\section{A geometric picture}\label{sec:geometry}
% =============================================================================

We now formulate a geometric picture, closely related to the connection between triorthogonality and anomaly cancellation~\cite{prakashGeneration}, of triorthogonality and the doubling procedure. This picture suggests a slightly more general doubling procedure than the one described above.

\subsection{Columns as charges}\label{sec:points}

Let $W\subseteq\F_p^n$ be a subspace with coefficients $\mathbf c$, and let $G$ be a generator matrix
with rows $g^0,g^1,\dots,g^r$. Write its $l$-th column as $(y_l,a_l)$, with $y_l$ the entry of
$g^0$ and $a_l\in\F_p^r$ the entries of the other rows. Assume for now that every
$y_l\neq0$.\footnote{A column with $y_l=0$ is a point at infinity. The lemmas below still hold in
the homogeneous form $\sum_l\mathbf c_lh_l^{\otimes3}=0$ for a space and
$\sum_l\mathbf c_lh_l^{\otimes3}=\lambda e_0^{\otimes3}$, with $e_0=(1,0,\dots,0)$, for a code, where
$h_l$ is the $l$-th column; this is Definition~\ref{def:tri} restated. Some seeds used below, such as
the cyclic $[[19,1,3]]_3$ and $[[11,1,3]]_5$ codes, have such columns. For them the pictures are only
partly literal, but the algebraic proof of Theorem~\ref{thm:doubling} is unaffected.} We
place at the point
\[
x_l=a_l/y_l\in\F_p^r\qquad\text{the charge}\qquad \chi_l=\mathbf c_l\,y_l^3\in\F_p .
\]
The row $g^0$ plays the role of an affine chart; other choices of $g^0$ give other pictures of the
same space.

\begin{lemma}\label{lem:neutral}
$W$ is a triorthogonal space if and only if
\begin{equation}\label{eq:neutral}
\sum_l\chi_l\,f(x_l)=0\qquad\text{for every polynomial $f$ on $\F_p^r$ of degree at most $3$.}
\end{equation}
\end{lemma}

\begin{proof}
Pulling $y_l^3$ out of each product gives
\[
\begin{aligned}
\tau(g^0,g^0,g^0)&=\textstyle\sum_l\chi_l, &
\tau(g^0,g^0,g^i)&=\textstyle\sum_l\chi_l\,x_{l,i},\\
\tau(g^0,g^i,g^j)&=\textstyle\sum_l\chi_l\,x_{l,i}x_{l,j}, &
\tau(g^i,g^j,g^k)&=\textstyle\sum_l\chi_l\,x_{l,i}x_{l,j}x_{l,k},
\end{aligned}
\]
and $\tau$ vanishes on $W$ exactly when all of these vanish, which is~\eqref{eq:neutral} for $f=1$
and for all monomials of degree $1$, $2$ and $3$.
\end{proof}

In electrostatic language, the charged points form a \emph{charge distribution}, and
Lemma~\ref{lem:neutral} says that its total charge and its dipole, quadrupole and octupole moments
vanish: the distribution is \emph{neutral up to octupole order}. The moments are the sums
$\sum_l\chi_l$, $\sum_l\chi_lx_{l,i}$, $\sum_l\chi_lx_{l,i}x_{l,j}$ and
$\sum_l\chi_lx_{l,i}x_{l,j}x_{l,k}$. (Positions, charges and moments take values in $\F_p$, and the
moments are the full Cartesian tensors rather than their traceless parts.)

A triorthogonal matrix $(H,\mathbf c)$ with one logical row and logical
exponent $\lambda$ is a triorthogonal space punctured at one coordinate: appending to $H$ the
column $e_0=(1,0,\dots,0)$ with coefficient $-\lambda$ changes $\tau(h^0,h^0,h^0)$ by $-\lambda$ and no
other value of $\tau$, so the extended rows span a triorthogonal space, and conversely. Taking the
logical row as $g^0$, the restored coordinate is a point at the origin with charge $-\lambda$.
Removing it leaves a distribution whose moments are those of a point charge $\lambda$ at the origin:

\begin{lemma}\label{lem:moments}
Let $(H,\mathbf c)$ have logical row $h^0$, stabilizer rows $h^1,\dots,h^r$, and columns
$(y_l,a_l)$ with every $y_l\neq0$. Then $(H,\mathbf c)$ is triorthogonal if and only if
\begin{equation}\label{eq:moments}
\sum_l\chi_l\,f(x_l)=\lambda\,f(0)\qquad\text{for every polynomial $f$ on $\F_p^r$ of degree at most $3$,}
\end{equation}
with $\lambda\neq0$, in which case $\lambda$ is the logical exponent.
\end{lemma}

\begin{proof}
Apply Lemma~\ref{lem:neutral} to the space obtained by restoring the punctured coordinate, whose
point is the origin with charge $-\lambda$.
\end{proof}

So a triorthogonal space is a neutral charge distribution, and a triorthogonal code is a neutral
distribution with one point removed, which then looks, up to octupole order, like a point charge of the opposite sign at
the position of the removed point.

For a $P$-code the logical row is $\one$, so the points are the columns of
a generator matrix of $S$, each with charge $1$. The conditions $S\perp\one$ and
$S\subseteq S^\perp$ say that the dipole and quadrupole moments vanish, and the condition $n\not\equiv0 \pmod p$ says that
the total charge is nonzero. So a $P$-code looks like a point charge at the origin up to
\emph{quadrupole} order:
\begin{equation}\label{eq:Pmoments}
\sum_l f(x_l)=n\,f(0)\qquad\text{for every polynomial $f$ of degree at most $2$.}
\end{equation}

\begin{example}[Reed--Muller codes]
The points of the punctured Reed--Muller code are all nonzero points of $\F_p^m$, each with charge
$1$. Every polynomial of degree less than $m(p-1)$ sums to zero over $\F_p^m$, so
$\sum_{x\neq0}f(x)=-f(0)$ for such $f$: the point set looks like a charge $-1$ at the origin, and
all of $\F_p^m$, with unit charges, is a neutral distribution. For
$m=2$ and $p=3$ this covers all $f$ of degree at most $3$, giving the $[[8,1,2]]_3$ code.
\end{example}

The distance of the code also admits a geometric interpretation.

\begin{lemma}\label{lem:affinehull}
If every $y_l\neq0$, the distance of the code of $H$ is the smallest number of points $x_l$ whose
affine hull contains the origin.
\end{lemma}

\begin{proof}
A vector $v$ is a nontrivial $Z$-logical exactly when $\sum_lv_la_l=0$ and $\sum_lv_ly_l\neq0$. Then
$\alpha_l=v_ly_l/\sum_kv_ky_k$ satisfies $\sum_l\alpha_l=1$ and $\sum_l\alpha_lx_l=0$, so the origin is an affine
combination of the points in the support of $v$; and conversely.
\end{proof}

Thus $d\ge3$ means that no point is at the origin and no line through two points passes through
it. In the Reed--Muller example $x$ and $-x$ are both points, so $d=2$.

\subsection{Doubling as a double layer}\label{sec:prism}

There are two natural geometric constructions to obtain new neutral charge distributions from old ones. One is to superimpose two charge distributions that are neutral up to a given order, and the result remains neutral. If two points of opposite charge are superimposed at the same position, they cancel, reducing the number of charges. The second is a double layer construction that raises by one the order up to which a distribution is neutral, at the cost of two copies. 
\begin{lemma}\label{lem:doublelayer}
Let $\mathcal X$ be a charge distribution of points $x_l\in\F_p^r$ with charges $\chi_l$ that is neutral up to quadrupole
order, and let $\Gamma\in\F_p$. Then the distribution in $\F_p^{r+1}$ consisting of the points $(x_l,0)$ with
charges $-\Gamma\chi_l$ and the points $(x_l,1)$ with charges $+\Gamma\chi_l$ is neutral up to octupole order.
\end{lemma}

\begin{proof}
For a polynomial $f(x,z)$ of degree at most $3$ the new distribution gives
$\Gamma\sum_l\chi_l\big[f(x_l,1)-f(x_l,0)\big]$. The bracket is a finite difference in $z$, so it is a
polynomial in $x$ of degree at most $2$: terms of $f$ without $z$ cancel between the layers, and
every other term contains $z$, so its degree in $x$ is at most $2$. Its sum against the charges of $\mathcal X$ vanishes.
\end{proof}

We now show that we can re-interpret the doubling construction in terms of the above operations. Let us first translate the doubled matrix of Theorem~\ref{thm:doubling} to a charge distribution, with its punctured coordinate
restored, dividing each column by its entry in the logical row $(-\kappa\one,\kappa\one,y)$. The
coordinates are the base coordinates $x$ (one per row $(s,-s,0)$), a height $z$ (the linking row) and
seed coordinates $w$ (one per row $(0,0,t)$). With $x'_l=-x_l/\kappa$, $\Gamma=\eta\kappa^3$ and $N$ the total
charge of the $P$-code ($N=n_P$ when it is unweighted), the qudits become the following charges:
\begin{center}
\begin{tabular}{lll}
\toprule
qudit & point $(x,z,w)$ & charge\\
\midrule
first copy, coordinate $l$ & $(x'_l,\,0,\,0)$ & $-\Gamma$\\
second copy, coordinate $l$ & $(x'_l,\,1,\,0)$ & $+\Gamma$\\
seed, coordinate $j$ & $(0,\,1,\,t_j/y_j)$ & $\mathbf c_jy_j^3$\\
restored coordinate & $(0,\,0,\,0)$ & $-\lambda_T=\Gamma N$\\
\bottomrule
\end{tabular}
\end{center}

This distribution is the superposition of two neutral ones (Figure~\ref{fig:prism}). The first is a double layer.
The points $x'_l$ of the $P$-code, together with the origin with charge $-N$, form a distribution neutral up
to quadrupole order, by~\eqref{eq:Pmoments}, and by Lemma~\ref{lem:doublelayer} its double layer is
neutral up to octupole order. It consists of the two copies, the origin with charge $\Gamma N$, and the
point $Q=(0,1,0)$ with charge $-\Gamma N$. The second is the triorthogonal seed. With its punctured coordinate
restored, the seed is a triorthogonal space, so by Lemma~\ref{lem:neutral} its points $t_j/y_j$,
together with the origin with charge $-\lambda_T$, form a neutral distribution; placed in the slice $x=0$,
$z=1$, its origin is at $Q$.
The two distributions meet at $Q$, where the charges are $-\Gamma N$ and $-\lambda_T$. The
condition~\eqref{eq:kappa} on $\kappa$ is exactly $\Gamma N=-\lambda_T$, so these cancel and $Q$ disappears.

The result is a distribution neutral up to octupole order, that is, a triorthogonal space by
Lemma~\ref{lem:neutral}, whose point at the origin has charge $\Gamma N=-\lambda_T$. Puncturing it there
leaves a point charge $\lambda_T$ at the origin: the doubled code, with the logical exponent of the
seed. Thus the $P$-code controls the moments up to quadrupole order, the opposite layers raise
this to octupole order, and the seed absorbs the charge left at $Q$.

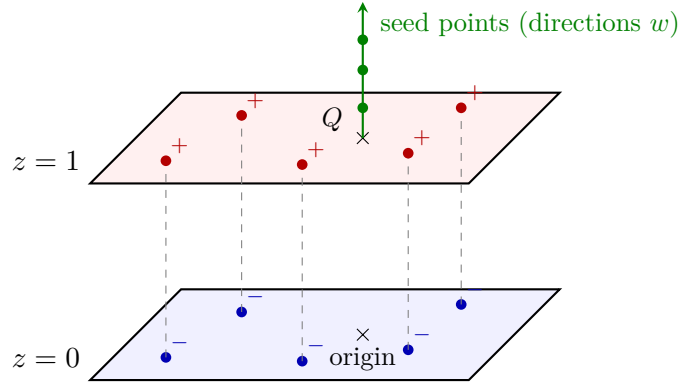
\begin{figure}[htbp]
\centering
\begin{tikzpicture}[x=1cm,y=1cm,>=stealth]
  \draw[thick,fill=blue!6] (0,0) -- (5,0) -- (6.2,1.2) -- (1.2,1.2) -- cycle;
  \draw[thick,fill=red!6]  (0,2.6) -- (5,2.6) -- (6.2,3.8) -- (1.2,3.8) -- cycle;
  \node[left] at (0,0.3) {$z=0$};
  \node[left] at (0,2.9) {$z=1$};
  \foreach \px/\py in {1.0/0.3, 2.0/0.9, 4.2/0.4, 4.9/1.0, 2.8/0.25}{
    \fill[blue!70!black] (\px,\py) circle (2pt);
    \node[blue!70!black,font=\scriptsize] at (\px+0.18,\py+0.2) {$-$};
    \draw[gray,dashed] (\px,\py) -- (\px,\py+2.6);
    \fill[red!70!black] (\px,\py+2.6) circle (2pt);
    \node[red!70!black,font=\scriptsize] at (\px+0.18,\py+2.8) {$+$};
  }
  \draw (3.6,0.6) node[cross out,draw,inner sep=2pt] {};
  \node[below,font=\small] at (3.6,0.55) {origin};
  \draw (3.6,3.2) node[cross out,draw,inner sep=2pt] {};
  \node[font=\small] at (3.2,3.45) {$Q$};
  \draw[->,thick,green!50!black] (3.6,3.2) -- (3.6,5.0);
  \foreach \sy in {3.6,4.1,4.5}{\fill[green!50!black] (3.6,\sy) circle (2pt);}
  \node[right,green!50!black,font=\small] at (3.7,4.7) {seed points (directions $w$)};
\end{tikzpicture}
\caption{The doubled code as a double layer. The points of the $P$-code, together with the origin,
appear in two parallel layers with opposite charges, and the points of the seed sit above $Q$. At
$Q$ the double layer carries charge $-\Gamma N$ and the restored point of the seed carries $-\lambda_T$;
by~\eqref{eq:kappa} these cancel. Puncturing the origin, which carries charge $\Gamma N=-\lambda_T$, gives
the doubled code.}
\label{fig:prism}
\end{figure}

In the language of Lemma~\ref{lem:affinehull}, the two cheapest ways to
reach the origin are via a $Z$-logical of the $P$-code within the bottom layer ($d_P$ points), or a
vertical pair of points together with seed points whose affine hull contains $Q$ ($2+d_T$
points), which is the formula of Theorem~\ref{thm:doubling}.

A punctured Reed--Muller code in one more dimension consists of
all $p$ parallel layers $z=0,\dots,p-1$, which raises the order by $p-1$ at the cost of $p$ layers.
Doubling raises it by one with two layers, and among tensor-product constructions this is the most
economical choice (Proposition~\ref{prop:optimal}).

\subsection{Weighted $P$-codes}\label{sec:masses}

The moment conditions see only the charges. Two qudits at the same point with charges $\chi$ and
$\chi'$ act exactly like one qudit with charge $\chi+\chi'$, and merging them does not lower the
distance, since Lemma~\ref{lem:affinehull} counts points. So a repeated column can be replaced by
a single weighted one.

\begin{definition}\label{def:wpcode}
A \emph{weighted $P$-code} is a pair $(S,\mu)$ with $S\subseteq\F_p^n$ and
$\mu\in(\F_p^\times)^n$ such that $\sum_l\mu_ls_ls'_l=0$ and $\sum_l\mu_ls_l=0$ for all
$s,s'\in S$, and $N:=\sum_l\mu_l\neq0$. Its distance is defined as in
Definition~\ref{def:pcode}.
\end{definition}

Its points are the columns of a generator matrix of $S$, with charges $\mu_l$. The argument
above never required that the charges equal $1$. Theorem~\ref{thm:doubling} therefore holds for
weighted $P$-codes, with coefficients $\eta\,\mu_l$ on both copies and $n_P$ replaced by $N$
in~\eqref{eq:kappa}. In the double layer of Section~\ref{sec:prism} the two copies then carry
charges $-\Gamma\mu_l$ and $+\Gamma\mu_l$.

The bound $n\ge2\dim S+1$ also still holds, since $S$ is totally-isotropic for
the form defined by $\sum_l\mu_lx_l^2$ in the nondegenerate space $\{v:\sum_l\mu_lv_l=0\}$. For $p=3$
the coefficients disappear from the doubled code: a coefficient $-1$ is removed by negating its column,
because $(-x)^3=-x^3$. Even so, weighted $P$-codes help for $p=3$: there are more of them than
unweighted ones, and the shortest is shorter, as illustrated by the following example.

\begin{example}[{$[[15,1,3]]_3$}]\label{ex:15}
The seven points of $\F_3^3$
\[
(0,0,1),\ (0,1,0),\ (0,1,1),\ (1,1,1),\ (1,1,2),\ (1,2,1),\ (1,2,2),
\]
with charges $(2,2,1,2,1,1,2)$, have total charge $11\equiv2 \pmod 3$ and vanishing dipole and quadrupole
moments, and no two of them are on a line through the origin. They form a weighted
$[[7,1,3]]^P_3$ code (Figure~\ref{fig:pcode7}). An exhaustive search shows that seven is the
minimum length of a weighted $P$-code with $d_P\ge3$, while ten is the minimum length of an
unweighted one.

Now double the trivial triorthogonal seed with it, which has $n_T=1$, $y=(1)$, $T_0=0$ and $\lambda_T=1$. Since
$N=11\equiv2$, \eqref{eq:kappa} becomes $\eta\cdot2\cdot\kappa^3+1\equiv0\pmod3$, which is
solved by $\eta=\kappa=1$. The matrix of Theorem~\ref{thm:doubling} then has five rows: the logical row
$(-\one,\one,1)$, the three rows $(s,-s,0)$ for the three rows $s$ of the matrix $S$ whose columns
are the seven points above, and the linking row $(0,\one,1)$. Its coordinates carry the
coefficients $\mu_l$ on both copies and $1$ on the last coordinate. Since $\mu_l=2\equiv-1$ for four
of the seven points, negating those columns in both copies makes every coefficient $1$, and gives
\[
H_{15}=\left(\begin{smallmatrix}
1&1&2&1&2&2&1&2&2&1&2&1&1&2&1\\ \hline
0&0&0&2&1&1&2&0&0&0&1&2&2&1&0\\
0&2&1&2&1&2&1&0&1&2&1&2&1&2&0\\
2&0&1&2&2&1&1&1&0&2&1&1&2&2&0\\
0&0&0&0&0&0&0&2&2&1&2&1&1&2&1
\end{smallmatrix}\right),
\]
a $[[15,1,3]]_3$ code on which $U^{\otimes15}$ implements logical $U$ exactly, with no Pauli
correction (Section~\ref{sec:doubled}). It is shorter than the smallest cyclic code,
$[[19,1,3]]_3$~\cite{PrakashCyclic}.

At distance $4$ and $5$, exhaustive searches find no weighted qutrit $P$-code of length at most $9$ or
$10$ respectively, so the unweighted codes of Table~\ref{tab:pcodes} remain the shortest there.
\end{example}

It was shown~\cite{prakashGeneration} that $[[15,1,3]]_3$ is, up to equivalence, the only triorthogonal qutrit code
of length $15$ and distance $3$, and no shorter one exists. It is also the reduction modulo $3$ of the
charge lattice of one generation of the Standard Model, with the right-handed neutrino removed, and
the coefficients $(2,2,1,2,1,1,2)$ of the $P$-code above are the right-handed isospins $2T_{3R}\bmod3$
of the corresponding fermions. Its description as a doubled code with this $P$-code also appears
in~\cite{prakashGeneration}.

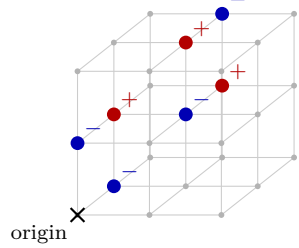
\begin{figure}[htbp]
\centering
\begin{tikzpicture}[x={(0.95cm,0cm)},y={(0.48cm,0.38cm)},z={(0cm,0.95cm)},>=stealth]
  \foreach \a in {0,1,2}{\foreach \b in {0,1,2}{\draw[gray!40] (\a,\b,0) -- (\a,\b,2);}}
  \foreach \a in {0,1,2}{\foreach \c in {0,1,2}{\draw[gray!40] (\a,0,\c) -- (\a,2,\c);}}
  \foreach \b in {0,1,2}{\foreach \c in {0,1,2}{\draw[gray!40] (0,\b,\c) -- (2,\b,\c);}}
  \foreach \a in {0,1,2}{\foreach \b in {0,1,2}{\foreach \c in {0,1,2}{\fill[gray!60] (\a,\b,\c) circle (1.1pt);}}}
  \foreach \a/\b/\c in {0/0/1,0/1/0,1/1/1,1/2/2}{
    \fill[blue!70!black] (\a,\b,\c) circle (2.6pt);
    \node[blue!70!black,font=\scriptsize] at ({\a+0.22},{\b},{\c+0.2}) {$-$};}
  \foreach \a/\b/\c in {0/1/1,1/1/2,1/2/1}{
    \fill[red!70!black] (\a,\b,\c) circle (2.6pt);
    \node[red!70!black,font=\scriptsize] at ({\a+0.22},{\b},{\c+0.2}) {$+$};}
  \draw[thick] (0,0,0) node[cross out,draw,inner sep=2.2pt] {};
  \node[below left,font=\scriptsize] at (0,0,0) {origin};
\end{tikzpicture}
\caption{The weighted $P$-code of Example~\ref{ex:15} in $\F_3^3$, with charges $+1$ (red) and $-1$
(blue). Lines over $\F_3$ wrap around, so collinearity is not visible in the drawing.}
\label{fig:pcode7}
\end{figure}

\subsection{Weighted $P$-codes from arcs}\label{sec:arcs}

For $p=3$ the charges are $\pm1$, and the short weighted $P$-codes above were found by
exhaustive search (Example~\ref{ex:15} and Section~\ref{sec:inputs}). For $p\ge5$ an explicit geometric
construction is available.    Of the codes it gives, only the $[[7,1,4]]^P_7$ code of Example~\ref{ex:20} is used in the doubled codes of Table~\ref{tab:best}, so readers interested only in the results can skip this section, apart from that example and the \emph{arc bound} $n\le p$ on the length of these codes.

The construction separates the problem into two parts. The \emph{directions} of the
points, that is, the lines through the origin on which they lie, determine the distance; the
\emph{positions} of the points on these lines, together with their \emph{charges}, determine the
moments. Taking all directions on the rational normal curve allows both parts to be solved in
closed form. 

Throughout, $r=\dim S$, and $\xi_1,\dots,\xi_n\in\F_p^r$ are the columns of a generator matrix
of $S$, with charges $\mu_l$. If $s^i$ denotes the $i$-th row of this matrix, the conditions
$\sum_l\mu_ls^i_l=0$ and $\sum_l\mu_ls^i_ls^j_l=0$ of Definition~\ref{def:wpcode} are the
entries of the dipole and quadrupole moments,
\begin{equation}\label{eq:wmoments}
\sum_l\mu_l\xi_l=0,\qquad \sum_l\mu_l\,\xi_l\xi_l^{\mathsf T}=0,
\end{equation}
and the total charge $N=\sum_l\mu_l$ must also be nonzero.

A weighted $P$-code satisfies two bounds. The first, $n\ge2r+1$, was shown
in Section~\ref{sec:masses}. The second is $d_P\le r+1$. In linear-algebraic form,
Lemma~\ref{lem:affinehull} states that $v$ is a nontrivial $Z$-logical exactly when
$\sum_lv_l(1,\xi_l)$ is a nonzero multiple of $e_0=(1,0,\dots,0)$. Hence $d_P$ is the smallest
number of columns of the $(r+1)\times n$ matrix whose first row is $\one$ and whose remaining
rows are those of the generator matrix, such that their span contains $e_0$. This matrix has
rank $r+1$, because the rows of the generator matrix are independent and $\one\notin S$; indeed,
taking $s=\one$ in Definition~\ref{def:wpcode} would give $N=\sum_l\mu_l=0$. Some $r+1$ of its
columns therefore form a basis of $\F_p^{r+1}$, and their span contains $e_0$. Together the two
bounds give $d_P\le r+1\le(n+1)/2$, and a code attaining both has parameters
$[[2r+1,1,r+1]]^P_p$.

To separate the roles of directions and positions, write each point as $\xi_l=\sigma_lv_l$, where $v_l$ is
a fixed representative of its direction $[v_l]\in\mathbb P^{r-1}(\F_p)$ and
$\sigma_l\in\F_p^\times$ is its position on that line. Suppose that the columns indexed by a set
$J$ span a space containing $e_0$, say $\sum_{l\in J}\alpha_l(1,\xi_l)=e_0$. The first entry
gives $\sum_{l\in J}\alpha_l=1$, so the $\alpha_l$ are not all zero, and the remaining entries
give $\sum_{l\in J}(\alpha_l\sigma_l)v_l=0$, a nontrivial linear dependency among the
directions. Consequently, if any $r$ of the directions are linearly independent, every such $J$
has at least $r+1$ elements, and $d_P=r+1$ for all choices of positions and charges. A set of
points of $\mathbb P^{r-1}(\F_p)$ with this property is called an \emph{arc}~\cite{Hirschfeld1998}; equivalently, the
matrix with columns $v_l$ generates an MDS code. The arc condition is sufficient but not
necessary: a dependency $\sum_lu_lv_l=0$ yields a $Z$-logical only if $\sum_lu_l/\sigma_l\neq0$, so
other directions can be used with suitable positions. The search of Appendix~\ref{app:search}
exploits this.

For the moments, substituting $\xi_l=\sigma_lv_l$
into~\eqref{eq:wmoments} gives
\begin{equation}
\sum_l\psi_l\,v_l=0,\qquad \sum_l\theta_l\,v_lv_l^{\mathsf T}=0,\qquad
\psi_l=\mu_l\sigma_l,\quad \theta_l=\mu_l\sigma_l^2 .
\end{equation}
For fixed directions both conditions are effectively linear: $\psi$ is a linear relation among the vectors
$v_l$, and $\theta$ is a linear relation among the symmetric matrices $v_lv_l^{\mathsf T}$.
Conversely, given such relations with non-zero entries, one can determine positions
$\sigma_l=\theta_l/\psi_l$ and charges $\mu_l=\psi_l^2/\theta_l$ that
satisfy~\eqref{eq:wmoments}. (The remaining condition, $N=\sum_l\psi_l^2/\theta_l\neq0$, still needs to be
checked separately.) Each relation has $n=2r+1$ unknowns. The dipole condition consists of $r$
equations, one for each entry of $v_l$, so a nonzero dipole relation always exists. The
quadrupole condition consists of one equation for each independent entry of a symmetric
$r\times r$ matrix, that is, $r(r+1)/2$ equations; for $r\ge4$ this exceeds $2r+1$, and for
generic directions the only solution is $\theta=0$. The directions must therefore be chosen so
that the matrices $v_lv_l^{\mathsf T}$ satisfy many linear relations amongst themselves.

Arcs cannot be arbitrarily long. For $r=2$ the directions are lines through the origin of $\F_p^2$, of
which there are $p+1$, and any two distinct lines are independent; an arc in $\mathbb
P^1(\F_p)$ therefore has at most $p+1$ points. The same bound holds in every dimension: for prime
$p$ and $2 \le r \le p$, an arc in $\mathbb P^{r-1}(\F_p)$ has at most $p+1$ points, or equivalently an
MDS code over $\F_p$ of dimension $r$ has length at most $p+1$. This is the MDS conjecture,
proved for prime fields in~\cite{Ball2012}. Since $n=2r+1$ is odd and $p+1$ is even, it gives
$n\le p$, hence $r\le(p-1)/2$ and $d_P\le(p+1)/2$; we refer to this as the \emph{arc bound}. It
permits $d_P=3$ at $n=5$ for $p=5$ and $d_P=4$ at $n=7$ for $p=7$. For $p=3$, however, $r \leq 1$,
which is why we cannot use the arc construction for qutrits.

\paragraph{The rational normal curve.} The standard example of an arc of maximal length is the
\emph{rational normal curve}~\cite{Hirschfeld1998}, consisting of the points
\[
v(t)=(1,t,t^2,\dots,t^{r-1}),\qquad t\in\F_p,
\]
of $\mathbb P^{r-1}(\F_p)$, together with the point $(0,\dots,0,1)$ at infinity, which is not
used here. For $r=2$ it is all of $\mathbb P^1(\F_p)$, and for $r=3$ it is a nondegenerate
conic; conversely, every nondegenerate conic in $\mathbb P^2(\F_p)$ takes this form after a
linear change of coordinates. A matrix with columns $v(t_1),\dots,v(t_n)$, for distinct $t_l$,
generates a Reed--Solomon code. From now on all directions are taken on this curve, which has
two relevant properties.

First, it is an arc. The point $v(t)$ lies on the hyperplane $\nu\cdot v=0$ exactly when the
polynomial $\nu_0+\nu_1t+\dots+\nu_{r-1}t^{r-1}$ vanishes at $t$, and a nonzero polynomial of
degree at most $r-1$ has at most $r-1$ roots. Hence no $r$ points of the curve lie on a
hyperplane; equivalently, any $r$ of them form an invertible Vandermonde matrix.

Second, its moment conditions reduce to conditions on a line. The entries of $v(t)$ are
$t^0,\dots,t^{r-1}$, and the $(i,j)$ entry of $v(t)v(t)^{\mathsf T}$ is $t^{i+j}$, which takes
only the values $t^0,\dots,t^{2r-2}$. Hence $\psi$ and $\theta$ are relations exactly when
\begin{equation}\label{eq:linerelations}
\sum_l\psi_l\,f(t_l)=0\ \ \text{for}\ \deg f\le r-1,\qquad
\sum_l\theta_l\,f(t_l)=0\ \ \text{for}\ \deg f\le 2r-2 .
\end{equation}
The quadrupole condition thus consists of $2r-1$ equations rather than $r(r+1)/2$, fewer than
the $2r+1$ unknowns. In the language of Section~\ref{sec:geometry}, $\psi$ and $\theta$ are
charge distributions on the line $\F_p$, supported at the points $t_l$, that are neutral up to
orders $r-1$ and $2r-2$ respectively.

Such distributions can be given explicitly. For distinct
$t_1,\dots,t_n\in\F_p$ put $\Delta_l=\prod_{m\neq l}(t_l-t_m)$. Lagrange interpolation writes
every polynomial $f$ of degree at most $n-1$ as
$$f(t)=\sum_lf(t_l)\prod_{m\ne l}(t-t_m)/\Delta_l,$$ and comparing the coefficients of $t^{n-1}$
on the two sides gives
\begin{equation}\label{eq:lagrange}
\sum_l\frac{f(t_l)}{\Delta_l}=\big[t^{n-1}\big]f .
\end{equation}
Thus the charges $1/\Delta_l$, all of which are nonzero, annihilate every polynomial of degree at
most $n-2$. In coding-theoretic terms, this is the classical fact that the dual of a
Reed--Solomon code is a generalized Reed--Solomon code with column multipliers $1/\Delta_l$~\cite{macwilliams77}.

\begin{proposition}\label{prop:rnc}
Let $p\ge5$, $1\le r\le(p-1)/2$ and $n=2r+1$. Let $t_1,\dots,t_n\in\F_p$ be distinct, let
$c\in\F_p^\times$, and let $q\in\F_p[t]$ have degree $r$, leading coefficient
$\operatorname{lc}(q)$, and $q(t_l)\neq0$ for all $l$. Then the points
$\xi_l=v(t_l)/q(t_l)$, that is, the columns of the $r\times n$ matrix with entries
$t_l^i/q(t_l)$ for $0\le i\le r-1$, with charges $\mu_l=c\,q(t_l)^2/\Delta_l$, form a weighted
$[[2r+1,1,r+1]]^P_p$ code with total charge $N=c\operatorname{lc}(q)^2$.
\end{proposition}

\begin{proof}
Take $\theta_l=c/\Delta_l$ and $\psi_l=c\,q(t_l)/\Delta_l$, which give
$\sigma_l=\theta_l/\psi_l=1/q(t_l)$ and $\mu_l=\psi_l^2/\theta_l=c\,q(t_l)^2/\Delta_l$. Since
$n-2=2r-1$, by~\eqref{eq:lagrange} $\theta$ annihilates every polynomial of degree at most
$2r-1$, which includes the quadrupole condition in~\eqref{eq:linerelations}. If $\deg f\le r-1$
then $\deg(qf)\le2r-1$, so $\sum_l\psi_lf(t_l)=c\sum_l(qf)(t_l)/\Delta_l=0$, which is the dipole
condition. Since $q^2$ has degree exactly $2r=n-1$,~\eqref{eq:lagrange} gives
$N=c\sum_lq(t_l)^2/\Delta_l=c\operatorname{lc}(q)^2\neq0$. The $r$ rows are independent: a
vanishing combination $\sum_ia_it_l^i/q(t_l)=0$ for all $l$ would give a polynomial
$\sum_ia_it^i$ of degree at most $r-1$ with $n>r-1$ roots. Finally, the directions $v(t_l)$ form
an arc, so $d_P=r+1$.
\end{proof}

A suitable $q$ always exists. If $n<p$, one may take $q=(t-b)^r$ with $b$ distinct from every
$t_l$. If $n=p$, every element of $\F_p$ is some $t_l$, so $q$ must have no roots in $\F_p$; in
that case $r=(p-1)/2\ge2$, and $q$ may be taken irreducible of degree $r$. In particular, a
weighted $[[p,1,(p+1)/2]]^P_p$ code exists for every prime $p\ge5$. In coding-theoretic terms, $S$ is the generalized Reed--Solomon code with
evaluation points $t_l$ and column multipliers $1/q(t_l)$, and the charges
$\mu_l$ are chosen so that its dual with respect to $\sum_l\mu_lx_ly_l$ is
$S+\F_p\one$, the generalized Reed--Solomon code of dimension $r+1$ with the
same points and multipliers. Similar choices of
evaluation points and column multipliers are used to construct self-dual and self-orthogonal
MDS codes from generalized Reed--Solomon
codes~\cite{GrasslGulliver2008,JinXing2017,Yan2019,FangLiuLuo2021}.

\begin{example}[{$[[5,1,3]]^P_5$}]\label{ex:5}
Take $p=5$, $r=2$, $c=1$ and $t_l=0,1,2,3,4$, all of $\F_5$. Then $\Delta_l$ is the derivative
of $t^5-t$ at $t_l$, so $\Delta_l=-1$ for every $l$, and $\mu_l=-q(t_l)^2$. With $q(t)=t^2+2$,
which has no roots in $\F_5$, we have $q(t_l)=(2,3,1,1,3)$, and the proposition gives
\[
S_5=\left(\begin{smallmatrix}
3&2&1&1&2\\
0&2&2&3&3
\end{smallmatrix}\right),\qquad \mu=(1,1,4,4,1),\qquad N=1 .
\]
For instance, the first column is $(1,0)/2=(3,0)$ and the second is $(1,1)/3=(2,2)$. The dipole
and quadrupole moments are readily verified by hand, and since the five directions $(1,t_l)$ are
distinct points of $\mathbb P^1(\F_5)$, $d_P=3$. Doubling the trivial seed with it ($\eta=1$, $\kappa=-1$) gives an $[[11,1,3]]_5$ code, which is equivalent to the cyclic code of~\cite{PrakashCyclic}.
\end{example}

The proposition covers essentially every code on the curve.
By~\eqref{eq:lagrange}, $\psi_l=Q(t_l)/\Delta_l$ is a dipole relation for every
$Q$ of degree at most $r$, and $\theta_l=g(t_l)/\Delta_l$ is a quadrupole relation
for every $g$ of degree at most $1$; these spaces have dimensions $r+1=n-r$ and
$2=n-(2r-1)$, so they are all the relations. Every weighted $P$-code of length
$2r+1$ with directions on the curve therefore has $\sigma_l=g(t_l)/Q(t_l)$ and
$\mu_l=Q(t_l)^2/(g(t_l)\Delta_l)$, and the proposition is the case of constant
$g$. When $n=p$ the $t_l$ exhaust $\F_p$, so a nonconstant $g$ would vanish at one
of them, and the proposition gives every such code, including that of
Example~\ref{ex:20} once the unused eighth point of its conic is moved to
infinity.

\begin{example}[{$[[7,1,4]]^P_7$ and a $[[20,1,4]]_7$ triorthogonal code}]\label{ex:20}
The columns of
\[
S_7=\left(\begin{smallmatrix}
0&0&2&1&1&1&4\\
0&4&0&1&2&3&6\\
2&0&0&1&3&2&3
\end{smallmatrix}\right),\qquad \mu=(2,2,2,4,4,1,2),
\]
lie on the conic $3x_0x_1+3x_0x_2+x_1x_2=0$ and form a weighted $P$-code of distance $4$, with
total charge $N=3$; up to a linear change of coordinates, it is an instance of
Proposition~\ref{prop:rnc}. Seven is the minimum length for $p=7$: by the bounds above,
$d_P\ge4$ requires $\dim S\ge3$, and hence $n\ge7$. Doubling the $[[6,1,2]]_7$ Reed--Solomon
seed~\cite{Krishna_2019} with it gives a $[[20,1,4]]_7$ code. Its coefficients include
non-cubes, so, unlike for $p=3$, they cannot be removed by rescaling columns. It is shorter than
the $[[28,1,5]]_7$ code of~\cite{PrakashCyclic}, the shortest there with $d\geq 4$.
\end{example}

For $p=2$ the charges are trivial, and the picture reduces to the correspondence between
triorthogonal codes and Reed--Muller polynomials used by Nezami and Haah~\cite{Nezami_2022}; doubling
the Steane code with the trivial seed, for instance, gives the fifteen nonzero points of $\F_2^4$, the
$[[15,1,3]]$ code. It gives no new qubit codes, since nontrivial charges and long arcs, the two
ingredients that help for odd $p$, do not exist over $\F_2$.

% =============================================================================
\section{Codes}\label{sec:results}
% =============================================================================

\subsection{The input codes}\label{sec:inputs}

The doubled codes below use the $P$-codes of Table~\ref{tab:pcodes}. Besides the weighted codes of
Sections~\ref{sec:masses} and~\ref{sec:arcs}, they come from cyclic and quadratic-residue codes and from search.

Let $\gcd(n,p)=1$ and let $S$ be cyclic with spectral support
$\Sigma\subseteq\Zn$~\cite{PrakashCyclic}: if $\alpha$ is a primitive $n$-th root of unity in an extension of
$\F_p$, then $S$ is spanned by the vectors $(\alpha^{jl})_l$ with $j\in\Sigma$, after the usual change of
basis.\footnote{We take $\alpha$ to be a root of the first irreducible factor of the cyclotomic
polynomial $\Phi_n$ over $\F_p$ in the factorization computed by SymPy. Another choice changes
$\Sigma$ by a multiplier, and for the codes taken from~\cite{PrakashCyclic} this choice gives codes
equivalent to those listed there.} For $t\ge1$ we write $t\Sigma=\Sigma+\dots+\Sigma$ ($t$ terms)
for the sumset, not the dilation $\{ts:s\in\Sigma\}$.

\begin{lemma}\label{lem:cyclicP}
A cyclic code $S$ is a $P$-code if and only if $0\notin\Sigma\cup2\Sigma$. Its distance is the
minimum weight of a vector with nonzero coordinate sum in the cyclic code with spectral support
$\Zn\setminus(-\Sigma)$.
\end{lemma}
\begin{proof}
The code $S$ is orthogonal to $\one$ if and only if $0\notin\Sigma$, and it is self-orthogonal if and only if
$\one\perp S^{\langle2\rangle}$, which by~\cite[Lemma~2]{PrakashCyclic} holds if and only if
$0\notin2\Sigma$. Finally, $S^\perp$ has spectral support $\Zn\setminus(-\Sigma)$.
\end{proof}

This is the criterion of~\cite{PrakashCyclic} for cyclic triorthogonal spaces without the condition on
$3\Sigma$, so the search method there applies.

Let $q\equiv3\pmod4$ be a prime such that $p$ is a
quadratic residue modulo $q$, and let $R$ be the set of nonzero squares modulo $q$. Since $-1$
is a non-residue, $0\notin R\cup2R$, so the expurgated quadratic-residue code with spectral
support $R$ is a $P$-code. Its distance is the minimum odd-like weight of the quadratic-residue
code, and so satisfies the square-root bound $d_P^2-d_P+1\ge q$~\cite{macwilliams77}. For $p=3$,
$q=11$ gives the $[[11,1,5]]_3^P$ ternary Golay code~\cite{Golay1949,2020golay}. For $p=5$, $q=11$ gives
$[[11,1,5]]_5^P$.

We searched all cyclic $P$-codes using Lemma~\ref{lem:cyclicP}, for
$n\le50$ when $p=3$, $n\le39$ when $p=5$, and $n\le30$ when $p=7$, and for $p=3$ we searched
exhaustively at small lengths. A $P$-code of distance at least $3$ is determined by a multiset of
nonzero columns $\xi_l$, any two equal or non-proportional, with $\sum_l\xi_l=0$ and
$\sum_l\xi_l\xi_l^{\mathsf T}=0$ and $n\not\equiv0\pmod p$. Enumerating these shows that no unweighted qutrit $P$-code with
$d_P\ge3$ has length below $10$, so the $[[10,1,4]]^P_3$ code
\[
S_{10}=\left(\begin{smallmatrix}
1&0&0&0&2&1&0&1&2&2\\
0&1&0&0&2&2&1&0&1&2\\
0&0&1&0&2&2&2&1&0&1\\
0&0&0&1&2&1&2&2&1&0
\end{smallmatrix}\right)
\]
is the shortest unweighted qutrit $P$-code with $d_P\ge3$. With nontrivial coefficients the minimum is $7$
(Example~\ref{ex:15}). For $p=5,7$, beyond the arc bound $n\le p$ of
Section~\ref{sec:arcs}, structured random searches found no weighted $P$-code with $d_P\ge5$
and $n\le10$, nor, for $p=5$, with $d_P\ge6$ and $n\le13$; these searches are not exhaustive.
Table~\ref{tab:pcodes} lists the $P$-codes we use.

\begin{table}[htbp]
\centering\small
\begin{tabular}{@{}lll@{}}
\toprule
$p$ & code & source\\
\midrule
$3$ & $[[7,1,3]]^P_3$ & weighted, Example~\ref{ex:15}; shortest weighted with $d\ge3$\\
    & $[[10,1,4]]^P_3$ & exhaustive search; shortest unweighted with $d\ge3$\\
    & $[[11,1,5]]^P_3$ & ternary Golay / QR, $\Sigma=\{1,3,4,5,9\}$\\
    & $[[13,1,4]]^P_3$ & cyclic\\
    & $[[22,1,7]]^P_3$ & cyclic, $\Sigma=\{1,2,3,5,6,8,9,10,15,18\}$\\
    & $[[23,1,8]]^P_3$ & cyclic / QR\\
\midrule
$5$ & $[[5,1,3]]^P_5$ & weighted, Example~\ref{ex:5}\\
    & $[[8,1,3]]^P_5$ & cyclic, $\Sigma=\{1,2,5\}$\\
    & $[[11,1,5]]^P_5$ & QR\\
    & $[[19,1,7]]^P_5$ & cyclic, $\Sigma=\{1,4,5,6,7,9,11,16,17\}$\\
\midrule
$7$ & $[[6,1,3]]^P_7$ & cyclic, $\Sigma=\{4,5\}$\\
    & $[[7,1,4]]^P_7$ & weighted, on a conic (Example~\ref{ex:20}); shortest with $d\ge4$\\
    & $[[9,1,4]]^P_7$ & cyclic, $\Sigma=\{1,4,6,7\}$\\
    & $[[18,1,6]]^P_7$ & cyclic, $\Sigma=\{2,5,8,11,12,14,15,17\}$\\
\bottomrule
\end{tabular}
\caption{Short $P$-codes, with spectral supports as in Section~\ref{sec:inputs}.}
\label{tab:pcodes}
\end{table}

\subsection{Doubled codes}\label{sec:doubled}

Doubling can be iterated, since every doubled code is again a seed. Table~\ref{tab:best} collects
the resulting codes together with, for $p=3,5,7$ and each distance, the shortest code known to us
from any source. Generator matrices are in Appendix~\ref{app:matrices}, and every distance is
certified by an explicit minimum-weight $Z$-logical together with a lower bound from
Theorem~\ref{thm:doubling} or exhaustive search. At distance three all three entries are optimal by
Proposition~\ref{prop:shortest} (Section~\ref{sec:search}). At distance four, the
exhaustive search (Table~\ref{tab:search}) leaves a single open length for $p=5$: no code has
length at most $14$, so the concatenated $[[16,1,4]]_5$ code is shortest unless a code of length
$15$ exists. For $p=3$ and $p=7$ the gaps, between $20$ and $28$ and between $12$ and $20$, remain
open. Doubling gives the shortest known codes for $p=3$ at $d=3,4,5,6,8$, for $p=5$ at $d=3$ (a code equivalent to the cyclic one) and $d=5$, and
for $p=7$ at $d=4$; elsewhere cyclic codes or concatenations remain shorter.

\begin{table}[htbp]
\centering\small
\begin{tabular}{@{}cccc@{}}
\toprule
$d$ & $p=3$ & $p=5$ & $p=7$\\
\midrule
$2$ & $8$ (RM, cyclic) & $4$ (RS, cyclic) & $6$ (RS, cyclic)\\
$3$ & $\mathbf{15}^*$ ($1+2\times7^{\rm w}$; $19$) & $11^*$ (cyclic; $1+2\times5^{\rm w}$) & $7^*$ (cyclic)\\
$4$ & $\mathbf{28}$ ($8+2\times10$; $40$) & $16$ (concatenated) & $\mathbf{20}$ ($6+2\times7^{\rm w}$; $28$)\\
$5$ & $\mathbf{37}$ ($15+2\times11$; $79$) & $\mathbf{33}$ ($11+2\times11$; $44$) & $28$ (cyclic)\\
$6$ & $\mathbf{72}$ ($28+2\times22$; $79$) & $44$ (concatenated) & $29$ (cyclic)\\
$7$ & $79$ (cyclic) & $64$ (concatenated) & $49$ (cyclic)\\
$8$ & $\mathbf{118}$ ($72+2\times23$; $121$) & $64$ (concatenated) & $49$ (concatenated)\\
\bottomrule
\end{tabular}
\caption{Shortest known length of an $[[n,1,d']]_p$ triorthogonal code with $d'\ge d$. Bold: codes of
this paper, where $a+2\times b$ is doubling a seed of length $a$ with a $P$-code of length $b$,
${}^{\rm w}$ marks a weighted $P$-code, and the last number is the shortest length
in~\cite{PrakashCyclic}; the $[[15,1,3]]_3$ code itself is from~\cite{prakashGeneration}. $^*$: shortest possible (Proposition~\ref{prop:shortest}). Reed--Muller (RM),
Reed--Solomon (RS), cyclic and concatenated codes are from~\cite{CampbellAnwarBrowne,Krishna_2019,PrakashCyclic}.}
\label{tab:best}
\end{table}

\begin{table}[htbp]
\centering\small
\begin{tabular}{@{}lcccccc@{}}
\toprule
code & $r$ & $\lambda$ & $A_d$ & $\epsilon_{\rm out}$ & $\epsilon^*$ & $P_{\rm acc}(0.01)$\\
\midrule
$[[8,1,2]]_3$ (seed) & $2$ & $-1$ & $8$ & $2\,\epsilon^2$ & $0.211$ & $0.92$\\
$[[19,1,3]]_3$ (seed)~\cite{PrakashCyclic} & $5$ & $-1$ & $42$ & $5.25\,\epsilon^3$ & $0.226$ & $0.83$\\
$[[15,1,3]]_3$ & $4$ & $1$ & $50$ & $6.25\,\epsilon^3$ & $0.230$ & $0.86$\\
$[[28,1,4]]_3$ & $7$ & $-1$ & $230$ & $14.4\,\epsilon^4$ & $0.212$ & $0.75$\\
$[[40,1,4]]_3$~\cite{PrakashCyclic} & $10$ & $1$ & $260$ & $16.25\,\epsilon^4$ & $0.223$ & $0.67$\\
$[[37,1,5]]_3$ & $10$ & $1$ & $2002$ & $62.6\,\epsilon^5$ & $0.230$ & $0.69$\\
$[[72,1,6]]_3$ & $18$ & $-1$ & $5060$ & $79.1\,\epsilon^6$ & $0.217$ & $0.48$\\
$[[118,1,8]]_3$ & $30$ & $-1$ & $245410$ & $959\,\epsilon^8$ & -- & --\\
\midrule
$[[11,1,3]]_5$ (seed)~\cite{PrakashCyclic} & $3$ & $-1$ & $92$ & $1.44\,\epsilon^3$ & $0.346$ & $0.90$\\
$[[33,1,5]]_5$ & $9$ & $-1$ & $2552$ & $2.49\,\epsilon^5$ & $0.334$ & $0.72$\\
\midrule
$[[6,1,2]]_7$ (seed)~\cite{Krishna_2019} & $1$ & $-1$ & $90$ & $2.5\,\epsilon^2$ & $0.232$ & $0.94$\\
$[[20,1,4]]_7$ & $5$ & $-1$ & $1194$ & $0.921\,\epsilon^4$ & $0.376$ & $0.82$\\
\bottomrule
\end{tabular}
\caption{Properties of the seeds and doubled codes. $r$ is the number of stabilizer rows of the triorthogonal matrix $H$, that is, the number of
rows other than the logical row; $\lambda$ is the logical exponent, modulo $p$; $A_d$ is the number of
$Z$-logicals of minimum weight, counted as vectors; and $\epsilon_{\rm out}$ is the leading term of
the output error. Thresholds and acceptance probabilities are computed exactly from weight
enumerators, using the MacWilliams identity~\cite{macwilliams77}, as functions of the infidelity
$\epsilon$. For $[[118,1,8]]_3$, with $30$ stabilizer rows, the weight enumerators are too
expensive to compute, so its threshold and acceptance probability are left blank, and $A_8$ is
obtained as described in the text. The cyclic $[[40,1,4]]_3$ code
of~\cite{PrakashCyclic} is included for comparison.}
\label{tab:perf}
\end{table}

Table~\ref{tab:perf} lists the seeds and doubled codes with their distillation properties. To
leading order the output error is $\epsilon_{\rm out}\approx A_d\,\epsilon^d/(p-1)^d$, because each
nonzero error pattern on $d$ given qudits occurs with probability $(\epsilon/(p-1))^d$. For example,
$[[15,1,3]]_3$ has $25$ supports of weight three, each carrying $2$ of the $8$ nonzero error patterns
as logicals, so $A_3=50$ and $\epsilon_{\rm out}\approx50(\epsilon/2)^3=6.25\,\epsilon^3$.

For $[[118,1,8]]_3$ the weight enumerators are out of reach, but $A_d$ follows from the proof of
Theorem~\ref{thm:doubling}. A $Z$-logical of weight $d=\min(d_P,d_T+2)$ is either a minimum-weight $Z$-logical
$v$ of the $P$-code split between the two copies, with the part on the second copy summing to zero,
or a minimum-weight $Z$-logical of the seed together with equal single entries at the same position of
both copies. If no vector $ae_i-e_j$ with $i\ne j$ lies in $S^\perp$, as for all codes here, this
gives
\[
A_d=\sum_v\#\Big\{J\subseteq\operatorname{supp}v:\ \textstyle\sum_{l\in J}v_l=0\Big\}+n_P\,A_{d_T}(T),
\]
where the sum runs over the minimum-weight $Z$-logicals $v$ of the $P$-code and is present only if
$d=d_P$, and the second term is present only if $d=d_T+2$. This reproduces $A_d$ for every other
doubled code in Table~\ref{tab:perf}, and gives $A_8=129030+23\cdot5060=245410$ for
$[[118,1,8]]_3$.

The new codes are shorter at equal distance, with comparable thresholds:
$[[28,1,4]]_3$ is $30\%$ shorter than $[[40,1,4]]_3$, with a slightly smaller leading coefficient
and a slightly lower threshold, and $[[15,1,3]]_3$ is shorter than $[[19,1,3]]_3$, with a somewhat
larger leading coefficient and an essentially equal threshold.

For $p=3$ the transversal gate implements $U^\lambda$ in general only up to a Pauli
correction~\cite{PrakashCyclic}. The cube sum $\Phi(x)=\sum_l\mathbf c_lx_l^3\bmod9$ is well
defined on the stabilizer space, because $(x+3k)^3\equiv x^3\pmod9$, and no correction is needed
exactly when $\Phi$ vanishes on a basis of stabilizers~\cite{prakashGeneration}. By this
criterion $[[15,1,3]]_3$ and $[[28,1,4]]_3$ are exact, implementing $U$ and $U^{-1}$, as are the
seeds $[[8,1,2]]_3$ and $[[19,1,3]]_3$, while $[[37,1,5]]_3$, $[[72,1,6]]_3$ and
$[[118,1,8]]_3$ need a fixed Pauli correction. For $[[15,1,3]]_3$ exactness also follows from its
lift to an integral triorthogonal lattice~\cite{prakashGeneration}.

\subsection{Codes with several outputs}\label{sec:multi}

Each code above is a triorthogonal space of length $n_W=n+1$ punctured at one coordinate
(Section~\ref{sec:points}), and the space can also be punctured at $k$ coordinates. If the $k$
punctured columns are independent and the logical rows are chosen dual to them, the result is an
$[[n_W-k,k,d]]_p$ code on which the transversal gate acts as a product of single-qudit gates
$U^{-\mathbf c_s}$, so each round yields $k$ magic states. Table~\ref{tab:multi} lists the largest $k$ found by a randomized search for $d\ge2$
and $d\ge3$; all distances are exact. At distance two these codes do not compete with the
$[[9m-k,k,2]]_3$ codes of~\cite{Prakash:2025azi}. At distance three they improve slightly on
$[[15,1,3]]_3$, with yield parameter $\log_d(n/k)$~\cite{Bravyi_2012}, equal to $2.37$ for $[[27,2,3]]_3$ and $2.24$ for
$[[35,3,3]]_3$, against $2.46$.

No puncture set with $k\ge2$ gave $d\ge4$. For puncture sets that include a coordinate of the
two copies there is a structural reason. In a doubled space the two copies
of column $l$ of the $P$-code, $(-\kappa,s_l,0,0)$ and $(\kappa,-s_l,\kappa,0)$, sum to the same vector
for every $l$, so puncturing a coordinate of the copies leaves a $Z$-logical of weight three. 

\begin{table}[htbp]
\centering\small
\begin{tabular}{@{}lccc@{}}
\toprule
space of & length, dimension & $d\ge2$ & $d\ge3$\\
\midrule
$[[15,1,3]]_3$ & $16,\ 5$ & $[[13,3,2]]_3$ & $k=1$\\
$[[19,1,3]]_3$ & $20,\ 6$ & $[[16,4,2]]_3$ & $k=1$\\
$[[28,1,4]]_3$ & $29,\ 8$ & $[[24,5,2]]_3$ & $[[27,2,3]]_3$\\
$[[37,1,5]]_3$ & $38,\ 11$ & $[[30,8,2]]_3$ & $[[35,3,3]]_3$\\
$[[72,1,6]]_3$ & $73,\ 19$ & $[[58,15,2]]_3$ & $[[70,3,3]]_3$\\
$[[118,1,8]]_3$ & $119,\ 31$ & $[[93,26,2]]_3$ & $[[115,4,3]]_3$\\
$[[33,1,5]]_5$ & $34,\ 10$ & $[[25,9,2]]_5$ & $[[31,3,3]]_5$\\
$[[20,1,4]]_7$ & $21,\ 6$ & $[[16,5,2]]_7$ & $[[19,2,3]]_7$\\
\bottomrule
\end{tabular}
\caption{Multi-output codes obtained by puncturing the triorthogonal space of each code at $k$
coordinates, with the largest $k$ found for $d\ge2$ and $d\ge3$. The puncture sets were found by a
randomized search, so these are achievable rather than optimal parameters.}
\label{tab:multi}
\end{table}

% =============================================================================
\section{Shortest codes: exhaustive search}\label{sec:search}
% =============================================================================

We determined the shortest triorthogonal codes of distance three for $p=5,7$ by exhaustive
search, extending the $p=3$ result of~\cite{prakashGeneration}, which our search reproduces, and
searched for short codes of distance four. For qubits, triorthogonal
codes have been classified exhaustively at small lengths~\cite{Nezami_2022,WillsJainSingh2026}.
The method, a reduction to a normal form followed by a backtracking search, is described in
Appendix~\ref{app:search}, together with its validation.

\begin{proposition}\label{prop:shortest}
The shortest triorthogonal $[[n,k,d]]_p$ codes with $d\ge3$ have length $15$ for $p=3$, $11$ for
$p=5$ and $7$ for $p=7$, with or without non-unit coefficients. In particular the cyclic $[[11,1,3]]_5$
and $[[7,1,3]]_7$ codes of~\cite{PrakashCyclic} are as short as possible, and $[[15,1,3]]_3$ is the
unique shortest qutrit code~\cite{prakashGeneration}.
\end{proposition}

At distance four the search is not complete, but for $p=5$ it excludes every length up to $14$,
so the shortest $[[n,1,4]]_5$ code has length $15$ or $16$.

\begin{table}[htbp]
\centering\small
\begin{tabular}{@{}llll@{}}
\toprule
$p$ & coefficients & $d\ge3$ & $d\ge4$\\
\midrule
$3$ & $1$ (suffices) & none for $n\le14$; $n=15$ (unique) & none for $n\le19$\\
$5$ & $1$ (suffices) & none for $n\le10$; $n=11$ & none for $n\le14$\\
$7$ & arbitrary & none for $n\le6$; $n=7$ & none for $n\le11$\\
$7$ & $\pm1$, length $=$ cost & none for $n\le6$; $n=7$ & none for $n\le13$\\
\bottomrule
\end{tabular}
\caption{Results of the exhaustive search. For $p=3,5$ unit coefficients suffice. In the last row
only unit coefficients are allowed, and length counts qudits, so that a coefficient in the class $\pm k$
costs $k$ qudits (Appendix~\ref{app:search}). For $d\ge4$ the table lists the lengths up to which
the search is complete; longer lengths were not searched to completion.}
\label{tab:search}
\end{table}

% =============================================================================
\section{Distillation sequences}\label{sec:sequences}
% =============================================================================

The codes above have distance at least three but rate $1/n$. The
$[[9m-k,k,2]]_3$ codes of~\cite{Prakash:2025azi} have distance two but high rate, and we combine them, as the 15-qubit code is combined with higher-rate codes for
qubits~\cite{Bravyi_2012,meier2012magicstatedistillationfourqubitcode}. We follow the model
of~\cite[Section~5]{Prakash:2025azi}: Clifford operations are perfect and free, inputs carry twirled
depolarizing noise of strength $\delta$, each round may use a different code, there are at most
five rounds, and the cost of a sequence is
\[
C=\prod_i\frac{n_i}{k_i\,P_{{\rm acc},i}},
\]
where $P_{{\rm acc},i}$ is the acceptance probability of round $i$. As usual, the outputs of a multi-output
round are treated as independent, which can be arranged by never feeding two outputs of the same
block into one later block~\cite{meier2012magicstatedistillationfourqubitcode}. We compute the
output error and acceptance probability of every qutrit round exactly from weight enumerators.
We compare three sets of codes: the codes of~\cite{Prakash:2025azi} with $k=3m-2$ and
$m\le12$;\footnote{The output error of these codes, for $k>1$, is slightly smaller than estimated
in~\cite{Prakash:2025azi}, so their cost was slightly overestimated there; we use exact values
throughout.}
the same codes together with $[[15,1,3]]_3$, $[[28,1,4]]_3$ and $[[37,1,5]]_3$; and, for reference,
the qubit 15-to-1~\cite{MSD}, 10-to-2~\cite{meier2012magicstatedistillationfourqubitcode}, $[[3k+8,k,2]]$ codes with even $k\le40$~\cite{Bravyi_2012}, and a doubled $[[49,1,5]]$ code~\cite{BravyiCross2015}. For the $[[3k+8,k,2]]$ codes we use the leading-order output
error $(1+3k)\epsilon^2$ of~\cite{Bravyi_2012}, which we checked against an exact computation for
$k=2$.

The 15-qutrit and 15-qubit codes both take fifteen inputs to one
output at distance three, so they have the same yield parameter and differ only in constants. In the two conventions of Section~\ref{sec:magic}, at equal infidelity $\epsilon$, the qutrit code has
$\epsilon_{\rm out}\approx6.25\,\epsilon^3$ against $35\,\epsilon^3$, because on each of its $25$
weight-three logical supports only $2$ of the $8$ nonzero error patterns are logical; its threshold is
$\epsilon^*=0.230$ against $0.141$, and the two codes accept with the same probability, $1-15\epsilon$,
to leading order. At equal $\delta$, which is the more demanding comparison for qudits,
these become $2.78\,\delta^3$ against $8.75\,\delta^3$ and $\delta^*=0.345$
against $0.283$, with nearly equal acceptance probabilities ($0.905$ and $0.928$ at
$\delta=0.01$). In either convention the qutrit code performs about as well as or better than the
15-qubit code~\cite{MSD}.

\begin{figure}[t]
\centering
\includegraphics[width=\linewidth]{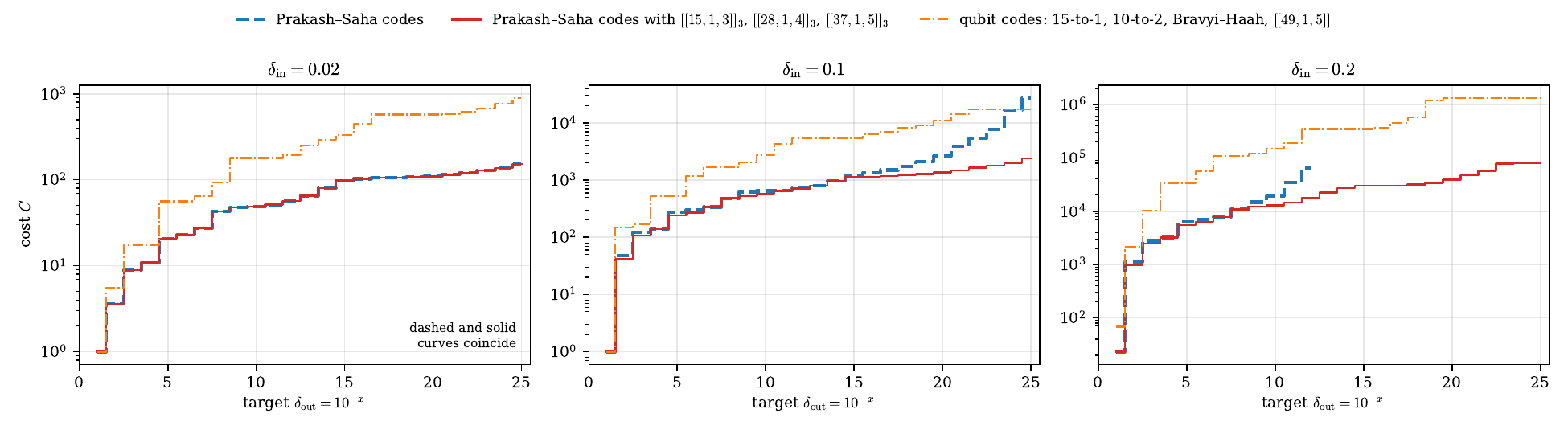}
\caption{Cost $C$ of the cheapest sequence of at most five rounds reaching
$\delta_{\rm out}\le10^{-x}$, for three input noise levels. Dashed: the Prakash--Saha codes~\cite{Prakash:2025azi}; solid: these together with $[[15,1,3]]_3$, $[[28,1,4]]_3$ and
$[[37,1,5]]_3$; dash-dotted: qubit codes. For $\delta_{\rm in}=0.02$ the dashed and solid curves
coincide. A missing segment means that the target is not reached within five rounds.}
\label{fig:cost}
\end{figure}

\begin{table}[t]
\centering\small\setlength{\tabcolsep}{5pt}
\begin{tabular}{@{}cccccl@{}}
\toprule
$\delta_{\rm in}$ & target & codes of~\cite{Prakash:2025azi} & $+$ new codes & qubits & sequence with new codes\\
\midrule
$0.02$ & $10^{-10}$ & $49$ & $49$ & $179$ & new codes not used\\
$0.02$ & $10^{-15}$ & $98$ & $98$ & $334$ & new codes not used\\
$0.02$ & $10^{-20}$ & $110$ & $110$ & $574$ & new codes not used\\
$0.02$ & $10^{-25}$ & $153$ & $153$ & $892$ & new codes not used\\
\midrule
$0.1$ & $10^{-10}$ & $648$ & $567$ & $2708$ & $15$, $26$, $68$, $68$\\
$0.1$ & $10^{-15}$ & $1182$ & $1140$ & $5486$ & $15$, $50$, $74$, $74$, $74$\\
$0.1$ & $10^{-20}$ & $2629$ & $1368$ & $10968$ & $15$, $20$, $50$, $74$, $74$\\
$0.1$ & $10^{-25}$ & $27194$ & $2407$ & $17101$ & $15$, $14$, $20$, $32$, $56$\\
\midrule
$0.2$ & $10^{-10}$ & $19139$ & $12978$ & $148654$ & $8$, $15$, $26$, $68$, $74$\\
$0.2$ & $10^{-15}$ & --- & $30050$ & $348392$ & $15$, $15$, $74$, $74$, $74$\\
$0.2$ & $10^{-20}$ & --- & $39158$ & $1323317$ & $15$, $15$, $26$, $38$, $68$\\
$0.2$ & $10^{-25}$ & --- & $81050$ & $1323317$ & $15$, $28$, $74$, $74$, $74$\\
\bottomrule
\end{tabular}
\caption{Costs of the optimal sequences of Figure~\ref{fig:cost}, computed exactly. Sequences list the
code lengths round by round; $15$ and $28$ are $[[15,1,3]]_3$ and $[[28,1,4]]_3$, and the others are the
Prakash--Saha codes $[[6m+2,3m-2,2]]_3$, which are determined by their length. A dash means that the target is not reached within five
rounds; equal costs mean that the same sequence reaches both targets.}
\label{tab:cost}
\end{table}

Figure~\ref{fig:cost} and Table~\ref{tab:cost} show the results. At
$\delta_{\rm in}=0.02$, the starting point of~\cite{Prakash:2025azi}, the high-rate distance-two
codes are already so efficient that the optimal sequences do not use the new codes, whose yield
parameter $\log15/\log3\approx2.46$ is larger. At higher input noise the optimal sequences use one or
two early rounds of $[[15,1,3]]_3$, sometimes followed by $[[28,1,4]]_3$, and continue with
high-rate codes: the new codes play the role that the 15-qubit code plays for qubits. The gain
grows with the target. At $\delta_{\rm in}=0.1$ the cost falls by $12\%$ at $10^{-10}$ and by a
factor of eleven at $10^{-25}$, and at $\delta_{\rm in}=0.2$ the distance-two codes alone do not
reach $10^{-15}$ within five rounds. In this model, qutrit sequences are cheaper than qubit sequences at all
three noise levels. The difference is already large at low input noise, where it is due to the
Prakash--Saha codes alone: at $\delta_{\rm in}=0.02$ they need between $3.4$ and $5.8$ times fewer
noisy inputs than the qubit codes, for example $153$ noisy qutrits against $892$ noisy qubits per
output at $10^{-25}$. This compares input counts at equal depolarizing strength, with the qubit
codes listed above. Noisy qutrit and qubit magic states need not cost the same to prepare, the
distilled gates differ, and qubit protocols with lower yield
parameter~\cite{Haah_2018,Hastings_2018} would narrow the gap.

% =============================================================================
\section{Conclusion}\label{sec:conclusion}
% =============================================================================

Doubling extends to every odd prime, with self-orthogonality required only of the $P$-code. A geometric picture of charge distributions  explains why doubling works and motivates the introduction of weighted $P$-codes. The resulting codes
lower the cost of qutrit distillation at high input noise, but for $p=5$ and $p=7$ several
shortest known codes remain cyclic or concatenated.

Better $P$-codes would immediately give better doubled codes. For $p=5,7$ this means weighted
$P$-codes beyond the arc bound $n\le p$, where our searches were only random; for $p=3$, codes derived from the extremal ternary self-dual codes of length
$36$ and $48$~\cite{macwilliams77}, for example by puncturing, are promising
candidates. It would also be interesting to extend doubling to codes with $k>1$ and, as for
qubits~\cite{Haah_2018}, to multi-qudit gates such as $CCZ$, and to include the cost of Clifford
operations in the comparison of distillation sequences.

\section*{Acknowledgements}
This work was supported by a DST-SERB grant (CRG/2021/009137) and by the DST-WISE KIRAN Ph.D.
fellowship (grant DST/\allowbreak WISE-PhD/\allowbreak PM/\allowbreak 2023/\allowbreak 32(G)). We thank T. Saha for sharing his unpublished notes on qudit CSS-T codes~\cite{SahaNotes}, in
which he showed, independently of~\cite{PrakashCyclic}, that for $p\ge5$ triorthogonality
without self-orthogonality suffices for a transversal non-Clifford gate. SP thanks Prof. P. S. Satsangi for discussions on qudit magic state distillation several years ago that inspired this work~\cite{DPS}.

We used large language models (Gemini 3.1, 3.6; Claude Sonnet 5.5 and Opus 5.5) to assist with
writing code, drafting text, and revision. We checked all results and take full responsibility for the content of this paper.

\section*{Code and data availability}
Code and data are available at
\url{https://github.com/spdei/qudit-triorthogonal-codes-from-doubling}. They include the
generator matrices of the doubled codes and their inputs in machine-readable form, scripts
that verify their triorthogonality and distances and compute the distillation results, and the
exhaustive-search program with the logs of its runs.

\bibliographystyle{quantum2}
\bibliography{qudit2026}

\appendix
\section{Tensor-product constructions}\label{app:blowup}

Doubling is the simplest member of a family of constructions in which every coordinate of a
small matrix is replaced by a copy of a $P$-code. This appendix defines the family
(Theorem~\ref{thm:blowup}), shows that doubling belongs to it, and shows that no other member
adds less length per unit of distance gained (Proposition~\ref{prop:optimal}).
Example~\ref{ex:tripling} gives the parameters of codes from two other members, for comparison
with doubling. The construction is reminiscent of, but distinct from, the inner and outer codes
of~\cite{Haah_2017}; as in the rest of the paper, we consider only codes encoding a
single qudit with large distance.

The construction has three ingredients.
\begin{itemize}
\item A $P$-code $S\subseteq\F_p^{n_P}$ with distance $d_P$.
\item A seed $T=(y;T_0)$ of length $n_T$ with distance $d_T$.
\item \emph{Inner data} on $m$ coordinates, consisting of block coefficients
$\mathbf c\in(\F_p^\times)^m$, a subspace $E\subseteq\F_p^m$, a \emph{logical vector}
$\rho\in\F_p^m$ and a \emph{linking vector} $u\in\F_p^m$.
\end{itemize}
For $a\in\F_p^m$ and $x\in\F_p^{n_P}$, let $a\otimes x\in\F_p^{mn_P}$ denote the vector whose
$i$-th block of $n_P$ coordinates is $a_ix$. The coordinates in block $i$ carry coefficient
$\mathbf c_i$, and the seed coordinates carry the seed's coefficients. The \emph{tensor-product construction}
$\mathcal B(S,T;E,\rho,u)$ is the matrix on $mn_P+n_T$ coordinates with the following rows:
\[
\begin{array}{lll}
\text{logical:} & (\rho\otimes\one,\ 0), &\\
\text{stabilizers:} & (\nu\otimes s,\ 0), & \nu\text{ in a basis of }E,\ s\text{ in a basis of }S,\\
& (u\otimes\one,\ y), &\\
& (0,\ t), & t\text{ in a basis of }T_0 .
\end{array}
\]
So each block is a copy of the $P$-code, the subspace $E$ says how the copies of $S$ are
combined, $\rho$ says how the logical row is spread over the blocks, and $u$ ties the blocks to
the logical row of the seed. For a weighted $P$-code the same holds with coefficients
$\mathbf c_i\mu_l$ in block $i$ and $n_P$ replaced by the total charge $N$ below.

\begin{theorem}[Tensor-product construction]\label{thm:blowup}
Suppose that
\begin{enumerate}
\item[(i)] $E$ is a triorthogonal space for $\tau_{\mathbf c}$;
\item[(ii)] $\tau_{\mathbf c}(\rho,\rho,u)=\tau_{\mathbf c}(\rho,u,u)=0$ and
$n_P\tau_{\mathbf c}(u,u,u)+\lambda_T=0$;
\item[(iii)] $\lambda:=n_P\tau_{\mathbf c}(\rho,\rho,\rho)\neq0$; and
\item[(iv)] the rows are linearly independent.
\end{enumerate}
Then $\mathcal B(S,T;E,\rho,u)$ is triorthogonal with logical exponent $\lambda$. Its distance $d$
satisfies $\min(d_P,d_{\rm in})\le d\le d_{\rm in}$, where
\[
d_{\rm in}=\min\Big\{\wt(\sigma)+\wt(z):\ \sigma\in E^\perp,\ z\in T_0^\perp,\
u\cdot\sigma+y\cdot z=0,\ \rho\cdot\sigma\neq0\Big\}.
\]
If some block $i$ has $\rho_i\neq0$ and $u_i=0$, then $d=\min(d_P,d_{\rm in})$.
\end{theorem}

\begin{proof}
\emph{Triorthogonality.} For $a,b,e\in\F_p^m$ and $x,x',x''\in\F_p^{n_P}$ the form $\tau$
factorizes:
\[
\tau\big(a\otimes x,\ b\otimes x',\ e\otimes x''\big)=\tau_{\mathbf c}(a,b,e)\,\tau(x,x',x'').
\]
The $P$-code identities $\tau(s,s',\one)=\beta(s,s')=0$, $\tau(s,\one,\one)=\sum_ls_l=0$ and
$\tau(\one,\one,\one)=n_P$ then dispose of the triples as follows.
\begin{itemize}
\item Three rows from $E\otimes S$: the factor $\tau_{\mathbf c}(\nu,\nu',\nu'')$ vanishes by~(i).
\item Two rows from $E\otimes S$ and one from $\{\rho,u\}\otimes\one$: the factor $\beta(s,s')$
vanishes.
\item One row from $E\otimes S$ and two from $\{\rho,u\}\otimes\one$: the factor $\sum_ls_l$
vanishes.
\item A row $(0,t)$ together with any rows: the only contributions come from the seed
coordinates. There the other rows restrict to $0$ or $y$, and $\tau(t,\cdot,\cdot)$ vanishes on
these by triorthogonality of $T$.
\item The remaining triples: $\tau(\rho,\rho,u)$ and $\tau(\rho,u,u)$ get no seed contribution
because $\rho$ has none, so they vanish by~(ii). The cube of the linking row is
$n_P\tau_{\mathbf c}(u,u,u)+\lambda_T=0$, again by~(ii).
\end{itemize}
The cube of the logical row is $\lambda\neq0$ by~(iii).

\emph{Distance.} Write a $Z$-type vector as $Z=(Z_1,\dots,Z_m,z)$ with $Z_i\in\F_p^{n_P}$, and
let $\sigma_i=\sum_lZ_{i,l}$. The vector $Z$ is a nontrivial $Z$-logical if and only if all of the
following hold:
\begin{itemize}
\item $\sum_i\nu_iZ_i\in S^\perp$ for every $\nu\in E$;
\item $z\in T_0^\perp$;
\item $u\cdot\sigma+y\cdot z=0$;
\item $\rho\cdot\sigma\neq0$.
\end{itemize}
There are two cases.
\begin{itemize}
\item If $\nu\cdot\sigma\neq0$ for some $\nu\in E$, then $v=\sum_i\nu_iZ_i$ lies in $S^\perp$
and $\sum_lv_l=\nu\cdot\sigma\neq0$. So $v$ is a $Z$-logical of the $P$-code, and
$\wt(Z)\ge\wt(v)\ge d_P$.
\item Otherwise $\sigma\in E^\perp$. Each nonzero $\sigma_i$ requires $Z_i\neq0$, so
$\wt(Z)\ge\wt(\sigma)+\wt(z)\ge d_{\rm in}$.
\end{itemize}
For the upper bound, take $\sigma,z$ attaining $d_{\rm in}$ and set $Z_i=\sigma_ie_1$. Then
$\sum_i\nu_iZ_i=(\nu\cdot\sigma)e_1=0$, so $Z=(\sigma\otimes e_1,z)$ is a $Z$-logical of weight
$d_{\rm in}$. If $\rho_i\neq0=u_i$, put a minimum-weight $Z$-logical $v$ of the $P$-code in block
$i$ and set $z=0$; then $\sigma=\sigma_ie_i$ with $\sigma_i\neq0$, so $u\cdot\sigma=0$ and
$\rho\cdot\sigma\neq0$, and this is a $Z$-logical of weight $d_P$.
\end{proof}

Theorem~\ref{thm:doubling} is the case
$m=2$, $\mathbf c=(\eta,\eta)$, $E=\F_p(1,-1)$, $\rho=(-\kappa,0)$ and $u=(0,\kappa)$, after
subtracting the linking row from the logical row of Theorem~\ref{thm:doubling}. Then
$E^\perp=\F_p(1,1)$, so $\sigma=(a,a)$ with $a\neq0$, which forces $y\cdot z\neq0$ and gives
$d_{\rm in}=d_T+2$.

To compare members of the family, we regard the $mn_P$ new coordinates as the \emph{cost} and the
increase of distance over $d_T$ as the \emph{gain}. Define
\[
g(E,\rho)=\min\{\wt(\sigma):\ \sigma\in E^\perp,\ \rho\cdot\sigma\neq0\},
\]
the smallest number of blocks that a $Z$-logical in the second case of the proof above, with
$\sigma\in E^\perp$, must touch. Since
$\tau_{\mathbf c}(\rho,\rho,\rho)\neq0$ while $\tau_{\mathbf c}$ vanishes on $E$, we have
$\rho\notin E$, so some $\sigma\in E^\perp$ has $\rho\cdot\sigma\ne0$ and $g$ is finite.
For doubling, $g=2$.

\begin{proposition}\label{prop:optimal}
In the setting of Theorem~\ref{thm:blowup}:
\begin{enumerate}
\item $d\le d_T+g$ and $g\le m$;
\item if $g=m$ then $m=2$;
\item if $\rho-\alpha u\notin E$ for every $\alpha\in\F_p$, then $d\le m$, whatever the seed.
\end{enumerate}
\end{proposition}

\begin{proof}
(1) Take $\sigma$ attaining $g$. If $u\cdot\sigma=0$ put $z=0$. Otherwise take a
minimum-weight $Z$-logical $z$ of $T$, rescaled so that $y\cdot z=-u\cdot\sigma$. Then
$(\sigma\otimes e_1,z)$ is a $Z$-logical of weight at most $g+d_T$. The bound $g\le m$ holds
because $\sigma$ has $m$ coordinates.

(2) For $m=1$ condition~(ii) cannot hold: $\rho_1\neq0$ by~(iii), so
$\tau_{\mathbf c}(\rho,\rho,u)=\mathbf c_1\rho_1^2u_1=0$ forces $u=0$, and then
$n_P\tau_{\mathbf c}(u,u,u)+\lambda_T=\lambda_T\neq0$. So assume $m\ge2$, let $V=E^\perp$, and
suppose $g=m$, so that every vector of $V\setminus\rho^\perp$ has full weight.
\begin{itemize}
\item \emph{$V$ is one-dimensional.} If $\dim V\ge2$, then $V\setminus\rho^\perp$ contains
independent vectors $\sigma_0,\sigma_1$, both of full weight. Since
$\rho\cdot(\sigma_0+t\sigma_1)=\rho\cdot\sigma_0+t\,\rho\cdot\sigma_1$ with
$\rho\cdot\sigma_0\neq0$, the line $\sigma_0+t\sigma_1$ leaves $V\setminus\rho^\perp$ for at
most one value of $t$. Coordinate $i$ vanishes at
$t_i=-\sigma_{0,i}/\sigma_{1,i}$, and the $t_i$ take at least two values, because $\sigma_0$ and
$\sigma_1$ are not proportional. So for some $i$ the vector $\sigma_0+t_i\sigma_1$ lies in
$V\setminus\rho^\perp$ and has a zero coordinate, a contradiction.
\item \emph{$E$ is a hyperplane.} Hence $V=\F_p\sigma_0$ with $\sigma_0$ of full weight, and
$E=\sigma_0^\perp$. Let $\ell_1,\dots,\ell_m\in\F_p^{m-1}$ be the columns of a generator matrix of
$E$. Their only linear relation is $\sum_i\sigma_{0,i}\ell_i=0$, and all its coefficients are
nonzero.
\item \emph{This forces $m=2$.} Condition~(i) says that
$\sum_i\mathbf c_i\,\ell_i\otimes\ell_i\otimes\ell_i=0$ in $(\F_p^{m-1})^{\otimes3}$. If $m\ge3$,
then $\ell_1,\dots,\ell_{m-1}$ is a basis and $\ell_m=\sum_ia_i\ell_i$ with every $a_i\ne0$. The
coefficient of $\ell_1\otimes\ell_1\otimes\ell_2$ in the sum is $\mathbf c_ma_1^2a_2\neq0$,
because no other term contributes to it. This is a contradiction.
\end{itemize}

(3) On $V=E^\perp$, the linear functions $\sigma\mapsto\rho\cdot\sigma$ and
$\sigma\mapsto u\cdot\sigma$ are not proportional: if $\rho\cdot\sigma=\alpha\,u\cdot\sigma$ for
all $\sigma\in V$, then $\rho-\alpha u\in V^\perp=E$. So the kernel of
$\sigma\mapsto u\cdot\sigma$ on $V$ is not contained in that of $\sigma\mapsto\rho\cdot\sigma$,
and there is a $\sigma\in V$ with $u\cdot\sigma=0$ and $\rho\cdot\sigma\neq0$,
and $(\sigma\otimes e_1,0)$ is a $Z$-logical of weight $\wt(\sigma)\le m$.
\end{proof}

For $m\ge3$, parts~(1) and~(2) give a gain of at most $g\le m-1$ at a cost of $mn_P$, so at least
$\frac{m}{m-1}n_P>n_P$ of added length per unit of distance. Doubling, when $d_P\ge d_T+2$, gains
$2$ at a cost of $2n_P$, which is exactly $n_P$ per unit. Part~(3) says that unless $u\in\alpha^{-1}\rho+E$ for some $\alpha$, the distance is at most $m$
however good the seed is; doubling satisfies this condition, since $\rho+u=-\kappa(1,-1)\in E$. These statements compare constructions that use the same $P$-code.
They do not exclude that a construction with $m\ge3$ is shorter than iterated doubling for some
target distance, since iterated doubling may need several different $P$-codes; Example~\ref{ex:tripling}
gives an instance in which iterated doubling is shorter.

\begin{example}[Tripling and beyond]\label{ex:tripling}
Take $p=5$. For $m=3$ and $\mathbf c=(1,1,1)$, the space $E=\F_5(1,1,2)$ is triorthogonal,
since $1+1+8\equiv0$. Take $\rho=(0,0,1)$ and $u=k(1,1,0)$. Then
$\tau_{\mathbf c}(\rho,\rho,u)=\tau_{\mathbf c}(\rho,u,u)=0$, because $\rho$ and $u$ have disjoint
supports, $\tau_{\mathbf c}(u,u,u)=2k^3$, so~(ii) fixes $k$, and $\lambda=n_P\neq0$. A vector
$\sigma\in E^\perp$ with $\rho\cdot\sigma=\sigma_3\ne0$ has weight at least $2$ and
$u\cdot\sigma=k(\sigma_1+\sigma_2)=-2k\sigma_3\ne0$, so it needs a $Z$-logical of the seed, and
$d_{\rm in}=d_T+2$. Since $\rho_3\neq0=u_3$, the ``tripled'' code is a
\[
[[\,3n_P+n_T,\ 1,\ \min(d_P,d_T+2)\,]]_5
\]
code: the distance of doubling, at the cost of a third copy of the $P$-code.

For $m=5$, the space $E$ spanned by $(1,0,1,1,1)$ and $(0,1,1,2,3)$ is triorthogonal for
$\mathbf c=(1,1,2,3,4)$. Its columns are distinct points of $\mathbb P^1(\F_5)$, so every
nonzero vector of $E^\perp$ has weight at least $3$ and $g=3$. With $\rho=(0,0,0,0,1)$ and
$u=k(0,1,1,2,0)$, condition~(ii) holds for a suitable $k$, $u\in\alpha^{-1}\rho+E$, and
$d_{\rm in}=d_T+3$; since $\rho_5\neq0=u_5$ this gives a
\[
[[\,5n_P+n_T,\ 1,\ \min(d_P,d_T+3)\,]]_5
\]
code. Unlike doubling, it can raise the distance by an odd amount.

For instance, with the seed $[[11,1,3]]_5$ and the $P$-code $[[19,1,7]]^P_5$, the $m=5$
construction gives a $[[106,1,6]]_5$ code, whereas doubling $[[33,1,5]]_5$ with the same
$P$-code gives the shorter $[[71,1,7]]_5$.
\end{example}

\section{The exhaustive search}\label{app:search}

This appendix describes the search behind Proposition~\ref{prop:shortest} and
Table~\ref{tab:search}, and, at the end, the reduction used to find weighted $P$-codes for $p\ge5$.

Rescaling coordinate $l$ by $a\in\F_p^\times$ multiplies its coefficient by
$a^{-3}$ and leaves $\tau$ unchanged, so only the class of a coefficient in
$\F_p^\times/(\F_p^\times)^3$ is relevant. For $p=3,5$ cubing is a bijection and all coefficients can be
set to $1$; for $p=7$ there are three classes, with representatives $1$, $2$, and $3$.

\begin{lemma}\label{lem:reduction}
If a triorthogonal $[[n,k,d]]_p$ code with $d\ge3$ exists, then one exists with $k=1$, length at
most $n$ and distance at least $d$, whose stabilizer space and logical row are, after a permutation
of coordinates,
\[
A=(I_r\,|\,B),\qquad h=(0\,|\,q),\qquad B\in\F_p^{r\times n_B},\ q\in\F_p^{n_B},
\]
with coefficients $c_i$ on the identity columns and $w_l$ on the columns of $B$, taken from the class
representatives, such that
\begin{enumerate}
\item[(a)] $c_i+\sum_lw_lB_{il}^3=0$ for all $i$;
\item[(b)] $\sum_lw_lB_{il}^2B_{jl}=0$ for all $i\ne j$;
\item[(c)] $\sum_lw_lB_{il}B_{jl}B_{kl}=0$ for all distinct $i,j,k$;
\item[(d)] $\lambda=\sum_lw_lq_l^3\ne0$;
\item[(e)] $\sum_lw_lq_l^2B_{il}=0$ for all $i$;
\item[(f)] $\sum_lw_lq_lB_{il}B_{jl}=0$ for all $i,j$;
\item[(g)] the columns of $B$ have Hamming weight at least two and are pairwise non-parallel.
\end{enumerate}
Moreover $r+1\le n_B\le(p^r-1)/(p-1)-r$.
\end{lemma}

\begin{proof}
Deleting all logical rows but one preserves triorthogonality and does not lower the distance,
since the new set $A^\perp\setminus H'^\perp$ is contained in the old one. Each column $g_l$ of
$H$ contributes $c_lg_l^{\otimes3}$ to $\tau$. A zero column can be deleted, and two columns with
$g_{l'}=a\,g_l$ can be replaced by the single column $g_l$ with coefficient $c_l+a^3c_{l'}$, or deleted
if this vanishes. These operations preserve every value of $\tau$, and they do not lower the
distance: a $Z$-logical of the new code, padded with zeros and rescaled at the merged coordinate, is a
$Z$-logical of the old one of the same weight. If the stabilizer rows become dependent they are replaced
by a basis of their span; the logical row stays outside it, since $\tau(h,h,h)=\lambda\ne0$ while
$\tau$ vanishes on stabilizers. When $H$ has no zero or parallel columns, $d\ge3$ if and only if $A$
is projective: a zero column $l$ of $A$ gives the $Z$-logical $e_l$, and columns of $A$ with
$a_{l'}=a\,a_l$ give the $Z$-logical $a\,e_l-e_{l'}$. Choosing a systematic basis of $A$ and adding
stabilizers to $h$ gives the stated form, and conditions (a)--(f) are the values of $\tau$ on
triples of the rows $x_i=(e_i\,|\,b_i)$ and $h$; (g) is projectivity of $A$. For the bounds, the
vectors $c*x_i*x_i$, $1\le i\le r$, are independent and lie in $H^\perp$, which has dimension
$n-r-1$; and the columns of $A$ are distinct points of $\mathrm{PG}(r-1,p)$.
\end{proof}

For $p=7$, a column whose coefficient lies in the class $\pm k$, $k\in\{1,2,3\}$, can be replaced by $k$ parallel
copies with unit coefficients, and no fewer such columns suffice. Replacing a column by copies does not
lower the distance, so the shortest code with unit coefficients and $d\ge3$ has length equal to the minimum of
$\sum_i\operatorname{cost}(c_i)+\sum_l\operatorname{cost}(w_l)$ over solutions of (a)--(g), where $\operatorname{cost}(\pm k)=k$.

The solution set is invariant under permuting the columns of $\binom wq$ and
$B$ together, permuting the rows $(c_i,b_i)$, and rescaling $h$ and the coefficients, which we use to fix
$\lambda=1$; for $p=7$ the columns of $\binom qB$ and the rows $b_i$ may also be multiplied by cube
roots of unity. Taking the lexicographically smallest element of each orbit, every orbit contains a
solution whose columns $(w;q;B)$ and rows $(c_i,b_i)$ are lexicographically non-decreasing and, for
$p=7$, whose columns and rows have first nonzero entry in $\{1,3\}$. A backtracking search over the
rows $b_1,\dots,b_r$, checking each condition at the step where its last row is chosen, therefore
covers every orbit. For $p\ge5$ the conditions on a new row are linear, quadratic or cubic in it, and
at each node the program enumerates the smaller of two candidate sets: the vectors compatible with
the ordering, or the solutions of the linear conditions. For $d\ge D$ with $D\ge4$, the search is the
same and every leaf is checked by brute force for $Z$-logicals of weight $3,\dots,D-1$; $Z$-logicals of
weight one and two are excluded by (g).

We validated the search in three ways. For small values of $r$ and $n_B$ it finds the
same equivalence classes as a brute-force enumeration of all solutions of (a)--(f),
with unit and with general coefficients. Every code it finds was rebuilt, checked
against Definition~\ref{def:tri}, and its distance confirmed by brute force. The
distance check used for $D\ge4$ was tested on eight codes with $3\le d\le5$,
including $[[28,1,4]]_3$, $[[33,1,5]]_5$ and $[[20,1,4]]_7$: in every case it
accepts for $D\le d$ and rejects for $D=d+1$.

\paragraph{Weighted $P$-codes for $p\ge5$.} Beyond the arc bound we searched as
follows. Fix the directions $[v_l]$. As in Section~\ref{sec:arcs}, the moment
conditions say that $\theta$ is a linear relation among the $v_lv_l^{\mathsf T}$
and $\psi$ one among the $v_l$, with $\sigma_l=\theta_l/\psi_l$ and
$\mu_l=\psi_l^2/\theta_l$. By Lemma~\ref{lem:affinehull}, $d_P\ge d$ fails
exactly when $\sum_{l\in J}\alpha_l\xi_l=0$ with $\sum_{l\in J}\alpha_l=1$ for
some $J$ with $|J|\le d-1$, that is, when some dependency
$\sum_{l\in J}u_lv_l=0$ has
$\sum_{l\in J}u_l/\sigma_l=\sum_{l\in J}u_l\psi_l/\theta_l\neq0$. For fixed
$\theta$ the distance condition is therefore linear in $\psi$: one condition
$\sum_{l\in J}u_l\psi_l/\theta_l=0$ for each such dependency. We choose $\theta$
in the solution space of the quadrupole conditions, then $\psi$ in that of the
dipole and distance conditions, and keep solutions with all entries nonzero and
$N\neq0$.

\input{appendix_codes}

\end{document}

%% file: appendix_codes.tex
\section{Generator matrices}\label{app:matrices}
All matrices are over $\F_p$. For a triorthogonal code the first row, above the line, is the logical row; for a $P$-code the rows are the stabilizers, and the logical row is $\one$. Coefficients are $1$ unless stated. Spectral supports are as in Table~\ref{tab:pcodes}. For doubled codes, the columns are ordered as the first copy of the $P$-code, the second copy, then the seed, and the rows as in Theorem~\ref{thm:doubling}. Every distance is certified by an explicit minimum-weight $Z$-logical, together with a lower bound from exact enumeration ($P$-codes), exhaustive search (seeds and the shorter doubled codes) or Theorem~\ref{thm:doubling}. We print the matrices of the codes used in the examples; the others are given by their construction, and all matrices and $Z$-logicals are in the repository in machine-readable form.

\subsection*{$P$-codes}
\paragraph{$[[10,1,4]]_{3}^{P}$.} Shortest unweighted qutrit $P$-code with $d_P\ge3$. Printed in Section~\ref{sec:inputs}.

\paragraph{$[[7,1,3]]_{3}^{P}$.} Weighted qutrit $P$-code of Example~\ref{ex:15}. Charges $(2,\allowbreak 2,\allowbreak 1,\allowbreak 2,\allowbreak 1,\allowbreak 1,\allowbreak 2)$.
\begin{center}\tiny\setlength{\arraycolsep}{1.2pt}$\left(\begin{array}{ccccccc}
0&0&0&1&1&1&1\\
0&1&1&1&1&2&2\\
1&0&1&1&2&1&2\\
\end{array}\right)$
\end{center}

\paragraph{$[[11,1,5]]_{3}^{P}$.} Cyclic, $\Sigma=\{1,3,4,5,9\}$; the ternary Golay code. Matrix in the repository.

\paragraph{$[[13,1,4]]_{3}^{P}$.} Cyclic, $\Sigma=\{2,4,5,6,10,12\}$. Matrix in the repository.

\paragraph{$[[22,1,7]]_{3}^{P}$.} Cyclic, $\Sigma=\{1,2,3,5,6,8,9,10,15,18\}$. Matrix in the repository.

\paragraph{$[[23,1,8]]_{3}^{P}$.} Cyclic, $\Sigma=\{5,7,10,11,14,15,17,19,20,21,22\}$; a quadratic-residue code. Matrix in the repository.

\paragraph{$[[8,1,3]]_{5}^{P}$.} Cyclic, $\Sigma=\{1,2,5\}$. Matrix in the repository.

\paragraph{$[[5,1,3]]_{5}^{P}$.} Weighted; the code $S_5$ of Example~\ref{ex:5}, printed there, with charges $(1,\allowbreak 1,\allowbreak 4,\allowbreak 4,\allowbreak 1)$.

\paragraph{$[[11,1,5]]_{5}^{P}$.} Cyclic, $\Sigma=\{1,3,4,5,9\}$; a quadratic-residue code. Matrix in the repository.

\paragraph{$[[19,1,7]]_{5}^{P}$.} Cyclic, $\Sigma=\{1,4,5,6,7,9,11,16,17\}$. Matrix in the repository.

\paragraph{$[[7,1,4]]_{7}^{P}$.} Weighted; seven points of a conic (Example~\ref{ex:20}). Charges $(2,\allowbreak 2,\allowbreak 2,\allowbreak 4,\allowbreak 4,\allowbreak 1,\allowbreak 2)$.
\begin{center}\tiny\setlength{\arraycolsep}{1.2pt}$\left(\begin{array}{ccccccc}
0&0&2&1&1&1&4\\
0&4&0&1&2&3&6\\
2&0&0&1&3&2&3\\
\end{array}\right)$
\end{center}

\paragraph{$[[6,1,3]]_{7}^{P}$.} Cyclic, $\Sigma=\{4,5\}$. Matrix in the repository.

\paragraph{$[[9,1,4]]_{7}^{P}$.} Cyclic, $\Sigma=\{1,4,6,7\}$. Matrix in the repository.

\paragraph{$[[18,1,6]]_{7}^{P}$.} Cyclic, $\Sigma=\{2,5,8,11,12,14,15,17\}$. Matrix in the repository.

\subsection*{Seeds}
\paragraph{$[[8,1,2]]_{3}$.} The punctured Reed--Muller code~\cite{CampbellAnwarBrowne}.
\begin{center}\tiny\setlength{\arraycolsep}{1.2pt}$\left(\begin{array}{cccccccc}
1&1&1&1&1&1&1&1\\\hline
1&2&0&1&2&0&1&2\\
0&0&1&1&1&2&2&2\\
\end{array}\right)$
\end{center}

\paragraph{$[[19,1,3]]_{3}$.} The cyclic code of~\cite[Example~14]{PrakashCyclic}. Matrix in the repository.

\paragraph{$[[1,1,1]]_{3}$.} The trivial seed.
\begin{center}\tiny\setlength{\arraycolsep}{1.2pt}$\left(\begin{array}{c}
1\\\hline
\end{array}\right)$
\end{center}

\paragraph{$[[11,1,3]]_{5}$.} Cyclic construction B of~\cite{PrakashCyclic}: $n=12$, $\Sigma=\{2,3,9,10\}$, punctured at coordinate $0$. Matrix in the repository.

\paragraph{$[[6,1,2]]_{7}$.} Cyclic construction A of~\cite{PrakashCyclic}, $n=6$, $\Sigma=\{1\}$; a Reed--Solomon code.
\begin{center}\tiny\setlength{\arraycolsep}{1.2pt}$\left(\begin{array}{cccccc}
1&1&1&1&1&1\\\hline
3&2&6&4&5&1\\
\end{array}\right)$
\end{center}

\subsection*{Doubled codes}
\paragraph{$[[15,1,3]]_{3}$.} Theorem~\ref{thm:doubling} with seed $[[1,1,1]]_{3}$ and $P$-code $[[7,1,3]]_{3}^{P}$, $\kappa=1$, $\eta=1$. The columns with coefficient $2\equiv-1$ are then negated, which makes every coefficient $1$. Printed in Example~\ref{ex:15}.

\paragraph{$[[28,1,4]]_{3}$.} Theorem~\ref{thm:doubling} with seed $[[8,1,2]]_{3}$ and $P$-code $[[10,1,4]]_{3}^{P}$, $\kappa=1$, $\eta=1$.
\begin{center}\tiny\setlength{\arraycolsep}{1.2pt}$\left(\begin{array}{cccccccccccccccccccccccccccc}
2&2&2&2&2&2&2&2&2&2&1&1&1&1&1&1&1&1&1&1&1&1&1&1&1&1&1&1\\\hline
1&0&0&0&2&1&0&1&2&2&2&0&0&0&1&2&0&2&1&1&0&0&0&0&0&0&0&0\\
0&1&0&0&2&2&1&0&1&2&0&2&0&0&1&1&2&0&2&1&0&0&0&0&0&0&0&0\\
0&0&1&0&2&2&2&1&0&1&0&0&2&0&1&1&1&2&0&2&0&0&0&0&0&0&0&0\\
0&0&0&1&2&1&2&2&1&0&0&0&0&2&1&2&1&1&2&0&0&0&0&0&0&0&0&0\\
0&0&0&0&0&0&0&0&0&0&1&1&1&1&1&1&1&1&1&1&1&1&1&1&1&1&1&1\\
0&0&0&0&0&0&0&0&0&0&0&0&0&0&0&0&0&0&0&0&1&2&0&1&2&0&1&2\\
0&0&0&0&0&0&0&0&0&0&0&0&0&0&0&0&0&0&0&0&0&0&1&1&1&2&2&2\\
\end{array}\right)$
\end{center}

\paragraph{$[[37,1,5]]_{3}$.} Theorem~\ref{thm:doubling} with seed $[[15,1,3]]_{3}$ and $P$-code $[[11,1,5]]_{3}^{P}$, $\kappa=1$, $\eta=1$. Matrix in the repository.

\paragraph{$[[72,1,6]]_{3}$.} Theorem~\ref{thm:doubling} with seed $[[28,1,4]]_{3}$ and $P$-code $[[22,1,7]]_{3}^{P}$, $\kappa=1$, $\eta=1$. Matrix in the repository.

\paragraph{$[[118,1,8]]_{3}$.} Theorem~\ref{thm:doubling} with seed $[[72,1,6]]_{3}$ and $P$-code $[[23,1,8]]_{3}^{P}$, $\kappa=2$, $\eta=1$. Matrix in the repository.

\paragraph{$[[33,1,5]]_{5}$.} Theorem~\ref{thm:doubling} with seed $[[11,1,3]]_{5}$ and $P$-code $[[11,1,5]]_{5}^{P}$, $\kappa=1$, $\eta=1$. Matrix in the repository.

\paragraph{$[[20,1,4]]_{7}$.} Theorem~\ref{thm:doubling} with seed $[[6,1,2]]_{7}$ and $P$-code $[[7,1,4]]_{7}^{P}$, $\kappa=3$, $\eta=2$. Coefficients $(4,\allowbreak 4,\allowbreak 4,\allowbreak 1,\allowbreak 1,\allowbreak 2,\allowbreak 4,\allowbreak 4,\allowbreak 4,\allowbreak 4,\allowbreak 1,\allowbreak 1,\allowbreak 2,\allowbreak 4,\allowbreak 1,\allowbreak 1,\allowbreak 1,\allowbreak 1,\allowbreak 1,\allowbreak 1)$.
\begin{center}\tiny\setlength{\arraycolsep}{1.2pt}$\left(\begin{array}{cccccccccccccccccccc}
4&4&4&4&4&4&4&3&3&3&3&3&3&3&1&1&1&1&1&1\\\hline
0&0&2&1&1&1&4&0&0&5&6&6&6&3&0&0&0&0&0&0\\
0&4&0&1&2&3&6&0&3&0&6&5&4&1&0&0&0&0&0&0\\
2&0&0&1&3&2&3&5&0&0&6&4&5&4&0&0&0&0&0&0\\
0&0&0&0&0&0&0&3&3&3&3&3&3&3&1&1&1&1&1&1\\
0&0&0&0&0&0&0&0&0&0&0&0&0&0&3&2&6&4&5&1\\
\end{array}\right)$
\end{center}